\documentclass[english,9pt,shortpaper,twoside,web]{ieeecolor}
\usepackage[T1]{fontenc}
\usepackage[latin9]{inputenc}
\usepackage{amsmath}
\usepackage{amssymb}
\usepackage{graphicx}

\makeatletter
\usepackage{generic}
\usepackage{balance}
\usepackage{dsfont}

\usepackage{cite}
  \usepackage{times}   
\newtheorem{theorem}{Theorem}[section]

\newtheorem{remark}[theorem]{Remark}

\newtheorem{lemma}[theorem]{Lemma}

\newtheorem{aplemma}{Lemma}[section]
\newtheorem{approp}{Proposition}[section]

\usepackage[all]{xy}

\usepackage{balance}

\usepackage{tikz}
\usetikzlibrary{positioning, arrows, matrix}

\usepackage{psfrag}

\usepackage{dsfont}
\newcommand{\mathbbm}[1]{\mathds{#1}}

\newtheorem{assumption}{Assumption}[section]
\newtheorem{example}{Example}[section]

\usepackage{url}

\usepackage{tikz}
\usepackage{pgfplots}

\usepackage{url}
\usepackage{verbatim}

\usepackage{textcomp}
\def\BibTeX{{\rm B\kern-.05em{\sc i\kern-.025em b}\kern-.08em
    T\kern-.1667em\lower.7ex\hbox{E}\kern-.125emX}}
\usepackage{xcolor}

\makeatother

\usepackage{babel}
\begin{document}
\title{Effects of  Jamming Attacks on Wireless Networked Control Systems Under Disturbance} 
\author{Ahmet Cetinkaya, \IEEEmembership{Member, IEEE}, Hideaki Ishii, \IEEEmembership{Fellow, IEEE}, Tomohisa Hayakawa, \IEEEmembership{Member, IEEE}
\thanks{This work is supported by JST ERATO HASUO Metamathematics for Systems Design Project (No.\ JPMJER1603),   by JST CREST Grant  No.\ JPMJCR15K3, and  by JSPS KAKENHI under Grants JP20K14771 and JP18H01460.}
\thanks{Ahmet Cetinkaya is with the National Institute of Informatics, Tokyo, 101-8430, Japan.  (e-mail: {\mbox cetinkaya@nii.ac.jp}).}
\thanks{Hideaki Ishii is with the Department of Computer Science,  Tokyo Institute of Technology, Yokohama  226-8502, Japan. (e-mail: {\mbox ishii@c.titech.ac.jp}).}
\thanks{Tomohisa Hayakawa is with the Department of Systems and Control Engineering,  Tokyo Institute of Technology, Tokyo 152-8552, Japan. (e-mail: {\mbox hayakawa@sc.e.titech.ac.jp}).}}

\maketitle

\begin{abstract}

Jamming attacks on wireless networked control systems are investigated
for the scenarios where the system dynamics face exogenous disturbance.
In particular, the control input packets are assumed to be transmitted
from a controller to a remotely located linear plant over an insecure
wireless communication channel that is subject to jamming attacks.
The time-varying likelihood of transmission failures on this channel
depends on the power of the jamming interference signal emitted by
an attacker. We show that jamming attacks can prevent stability when
the system faces disturbance, even if the attacked system without
disturbance is stable. We also show that stability under jamming and
disturbance can be achieved if the average jamming interference power
is restricted in a certain way that we characterize in the paper.
We illustrate our results on an example networked control system with
a fading wireless channel, where the outage probability is affected
by jamming attacks. \end{abstract} 


\section{Introduction}

As the Internet of Things is gaining popularity, the use of wireless
communication channels and the Internet is increasing in remote control
applications. These communication technologies are easy to set up
and they provide efficiency in the transmission of measurement and
control data, but they can create major cyber-security issues in a
networked system \cite{wholejournal2015}. In the framework of cyber-physical
systems, researchers have identified a range of potential cyber attacks
with different properties \cite{chong2019tutorial,dibaji2019systems,lun2019state}.
For instance, an attacker who is knowledgeable about the system dynamics
can disrupt control operation by injecting false data into the system
or altering measurement and control data \cite{mo2010false,fawzi2014secure}.
Attackers with limited information can also cause cyber-security issues
by means of denial-of-service (DoS) attacks to prevent communication
over networks. For example, a jamming attacker can effectively prevent
transmission of packets over wireless channels by emitting sufficiently
strong interference signals, \cite{pelechrinis2011}. Jamming attacks
can result in performance issues and instability in wireless networked
control systems. 

The effects of jamming and other DoS attacks in control systems have
recently been investigated (see \cite{cetinkaya-entropy} for an overview).
In those works various attack models have been considered. For instance,
\cite{shishehsiam2016} considered a model where the attacker conducts
cycles of sleeping and jamming in a repetitive fashion. Moreover,
the works \cite{de2015inputtran,IFACde2016networked,cetinkaya2016tac,feng2017resilient}
considered models that allow the timing of attack strategies to be
arbitrary as long as the average attack duration and the average frequency
of attacks satisfy certain bounds. It was first observed in \cite{de2015inputtran}
that when a control system is subject to disturbance, duration and
frequency conditions for attacks need to be stronger to guarantee
stability in comparison to the case without disturbance. 

In this paper, our goal is to investigate the effects of jamming attacks
specifically for \emph{wireless} networked control problems that are
subject to disturbance. In particular, we consider the control problem
over a wireless channel, where the transmission failure model can
be characterized through the time-dependent Signal-to-Interference-plus-Noise-Ratio
(SINR), which is the ratio of the transmission power of the signal
to the jamming attacker's interference power summed with the channel
noise power. We consider channel models explored in the wireless communications
literature \cite{IFACproakis,molisch2011book}. In those models, the
effect of SINR on transmission failures is described through probabilistic
relations. A jamming signal with a strong interference power results
in a smaller SINR, which ends up increasing the likelihood of a transmission
failure. For instance, in wireless channels with fading, small SINR
increases the so-called outage probability, as explored in \cite{zhu2013networked,sheikholeslami2014jamming,hu2015distributed}.
Jamming can affect practical wireless communication networks. For
instance, effects of jamming attacks on SINR, packet decoding errors,
and failures in IEEE 802.11 communication networks are investigated
through experiments in \cite{fragkiadakis2014denial}. 

Previously, SINR-based channel models were used by \cite{IFACli2015jamming,IFACzhangattack,IFACli2016sinr}
for game-theoretic analysis of remote state estimation problems under
jamming attacks. Moreover, in \cite{IFACzhang2016optimal}, a probabilistic
channel model was considered in a networked control problem setting
and optimal attack policies were explored for the case where the total
number of attacks in a fixed interval is bounded. In \cite{ahmetifacwc2017},
we used an SINR-based probabilistic model to investigate a discrete-time
networked stabilization problem for scenarios where there is no disturbance,
but a jamming attacker can jam the wireless channel at each time instant
with a different interference power level that is unknown a priori.
Our results in \cite{ahmetifacwc2017} indicate that stabilization
can be achieved if the average interference power is bounded in the
long run even if the power can be very large at certain times.

In this paper we consider situations where the jamming attacker can
strategically change the interference power levels  at each time,
as in \cite{ahmetifacwc2017}. However, differently from \cite{ahmetifacwc2017},
we now consider disturbance, and through stochastic analysis, we show
that when the dynamics is subject to disturbance, jamming attacks
can potentially become more dangerous. Our results indicate that a
strategic attacker may take advantage of the disturbance to cause
instability even if the attacked system without disturbance is stable.
Specifically, the attacker can cause the state norm to grow to arbitrarily
large values with arbitrarily high probabilities, while keeping the
average jamming interference power below a threshold in the long run.
Thus, as in the deterministic case discussed in \cite{de2015inputtran},
a restriction is also needed in this paper. We consider a probabilistic
model and the attacker can only partially affect the occurrence probability
of a transmission failure. We show that when jamming attacks are restricted
so that the wireless channel is not subject to long consecutive emissions
of high powered interference signals, then the first moment of the
state stays bounded. Interestingly, even under such restrictions,
the wireless channel may be attacked at all time instants with small
interference powers and thus for any finite interval, there is always
a positive probability that all transmission attempts may fail. In
this aspect, our setting differs from the deterministic case, where
the maximum possible length of a continuous attack duration is required
to be bounded to ensure input-to-state stability under disturbance. 

As a first step, we investigate the scenarios where the norm of the
disturbance is bounded almost surely at each time by a fixed scalar.
In such scenarios, the first moment of the state is bounded under
attacks from an attacker with sufficiently small resources. Then we
explore the more general case where the distribution of the disturbance
norm may have infinite support. For this case, we obtain an inequality
for the first moment of the state that resembles those used for establishing
noise-to-state stability in stochastic systems (e.g., \cite{nunez2014,zhang2016noise}).
In particular, we obtain an upper bound of the first moment of the
state by utilizing the second moment of the disturbance. In our analysis,
a key technical role is played by a nondecreasing and concave function
of the attacker's interference power that upper-bounds the transmission
failure probability. In addition, the use of the first moment of the
state in the analysis facilitates the investigation of cross product
terms that involve the disturbance and the indicator process for transmission
failures through induced matrix norms. A practical consequence is
that our results can be used in the scenarios where the transmission
failures and the disturbance are statistically dependent. This is
for example the case when the disturbance is partially or fully caused
by attacker's actions and the jamming interference in the wireless
channel results in packet content errors. 

The paper is organized as follows. We explain the wireless networked
control problem under jamming attacks in Section~\ref{sec:Networked-Control-Under}.
In Section~\ref{sec:Analysis-of-Networked}, we explain the effects
of jamming attacks on systems with disturbance and present conditions
for stabilization. We present an example in Section~\ref{sec:Illustrative-Numerical-Examples}
to explore the effects of jamming attacks on the outage probability
of a wireless channel and the effects on the networked control system
that utilizes that particular channel. Finally, we conclude the paper
in Section~\ref{sec:Conclusion}. Our preliminary conference report
\cite{ahmetmtns2018} contains some of the results. In this paper,
we provide the proofs, additional detailed discussions, and a new
example. 

Throughout the paper, we use $\mathbb{N}$ and $\mathbb{N}_{0}$ to
denote the sets of positive and nonnegative integers, respectively.
Moreover, $\left\Vert \cdot\right\Vert _{2}$ denotes the Euclidean
norm, $\mathrm{\mathbb{P}}[\cdot]$ and $\mathbb{E}[\cdot]$ respectively
denote the probability and the expectation on a probability space
$(\Omega,\mathcal{F},\mathbb{P})$.  In the presentation of our stability
results, we use induced matrix norms (see Section 5.6 in \cite{hornmatrixanalysis}).
Specifically, for a given matrix $M\in\mathbb{R}^{n\times n}$, we
use $\|M\|$ to denote the induced matrix norm defined by $\|M\|\triangleq\sup_{x\in\mathbb{R}^{n}\setminus\{0\}}\frac{\|Mx\|}{\|x\|}$,
where $\|\cdot\|$ on the right-hand side denotes a vector norm on
$\mathbb{R}^{n}$. 

\section{Networked Control Under Jamming Attacks }

\label{sec:Networked-Control-Under}

We consider the networked control problem of a discrete-time linear
plant with a static state feedback controller. As illustrated in Fig.~\ref{Flo:NetworkedOperation},
a wireless communication channel is used for transmission of control
command  packets from the controller to the plant. This channel is
subject to transmission failures at certain times due to interference
caused by the jamming signal of an attacker, where the strength (or
power) can be tunable by the attacker. 

\begin{figure}
\centering  \includegraphics[width=0.8\columnwidth]{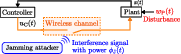} 

\caption{Operation of networked control system under jamming attacks }
 \label{Flo:NetworkedOperation} 
\end{figure}

In the networked control operation, at each time step $t$, the controller
computes a control command using the state information and attempts
to transmit it on the wireless channel. If the transmission is successful,
then the transmitted control command is applied at the plant side.
If, on the other hand, there is a transmission failure, then the control
input at the plant side is set to $0$. In this setting, the dynamics
of the plant is given by
\begin{align}
x(t+1) & =Ax(t)+(1-l(t))Bu_{\mathrm{C}}(t)+w_{\mathrm{P}}(t),\label{eq:system}
\end{align}
 where $x(t)\in\mathbb{R}^{n}$ is the state, $u_{\mathrm{C}}(t)\in\mathbb{R}^{m}$
is the control command  that is attempted to be transmitted by the
controller to the plant at time $t$, $w_{\mathrm{P}}(t)\in\mathbb{R}^{n}$
is the disturbance, and $l(t)\in\{0,1\}$ represents the transmission
status (with $l(t)=1$ indicating failure and $l(t)=0$ indicating
success). Moreover, $A\in\mathbb{R}^{n\times n}$ is the unstable
system matrix and $B^{n\times m}$ is the input matrix. 

In this paper, we investigate the networked stabilization of the plant
(\ref{eq:system}) through a state-feedback controller, where the
control command transmitted by the controller is given by 
\begin{align}
u_{\mathrm{C}}(t) & =Kx(t)+w_{\mathrm{C}}(t),\quad t\in\mathbb{N}_{0},\label{eq:control-input}
\end{align}
where $K\in\mathbb{R}^{m\times n}$ denotes the feedback gain, and
$w_{\mathrm{C}}(t)\in\mathbb{R}^{m}$ is used for describing disturbances
on the control command. 

\subsection{Closed-Loop System Dynamics}

With $w(t)\triangleq w_{\mathrm{P}}(t)+(1-l(t))Bw_{\mathrm{C}}(t)$,
the closed-loop networked control system (\ref{eq:system}), (\ref{eq:control-input})
becomes
\begin{align}
x(t+1)=Ax(t)+(1-l(t))BKx(t)+w(t),\,\,t\in\mathbb{N}_{0}.\label{eq:closed-loop-system}
\end{align}
The vector $w(t)$ in (\ref{eq:closed-loop-system}) represents the
\emph{overall disturbance} in the control system dynamics and it is
not related to the jamming signal emitted by the attacker. In our
problem setting, the jamming action affects the probability of successful/failed
delivery of control commands (as we will explain below more precisely).
In this sense, our problem setting is similar to those in \cite{de2015inputtran,IFACli2016sinr,IFACli2015jamming},
which involve DoS attacks causing packet losses. We note that there
are other problem settings in the literature, where the notion of
``jamming'' is used for describing the noise on the transmitted
data (see, e.g., \cite{basar1983gaussian,basar1985complete,akyol2015optimal});
there, the received data is the sum of the original data and the jamming
noise. In contrast, in our work, when there is a successful delivery,
the control command $u_{\mathrm{C}}(t)$, which is transmitted from
the controller, is assumed to be received by the plant (and applied
as an input) without any change. 

\begin{remark}In this paper, we present our results in terms of the
overall disturbance $w(t)$ in (\ref{eq:closed-loop-system}), which
includes exogenous disturbances on the plant modeled with $w_{\mathrm{P}}(t)$,
as well as potentially network-related disturbances on the controller
modeled with $w_{\mathrm{C}}(t)$. For the scenarios where the state
measurement is noisy, the effects of noise can also be represented
through the process $w_{\mathrm{C}}(t)$. In such cases, the control
command is given by $K\tilde{x}(t)$, where $\tilde{x}(t)=x(t)+\eta(t)$
is the measured state and $\eta(t)\in\mathbb{R}^{n}$ represents the
measurement noise. This situation is represented through (\ref{eq:control-input})
by setting $w_{\mathrm{C}}(t)\triangleq K\eta(t)$. We also note that
in our analysis, $w(t)$ is considered as a stochastic process. However,
we do not assume to know its distribution. \hfill \hfill $\triangleleft$
\end{remark}

\begin{remark}While we derived the closed-loop system (\ref{eq:closed-loop-system})
using a static state-feedback controller, the form of the dynamics
 in (\ref{eq:closed-loop-system}) also allows representing closed-loop
systems under other control architectures. For instance, one can consider
a dynamic controller 
\begin{align*}
x_{\mathrm{C}}(t+1) & =A_{\mathrm{C}}x_{\mathrm{C}}(t)+B_{\mathrm{C}}x(t),\\
u_{\mathrm{C}}(t) & =C_{\mathrm{C}}x_{\mathrm{C}}(t)+D_{\mathrm{C}}x(t)+w_{\mathrm{C}}(t),
\end{align*}
 where $x(t)$ is the state of plant (\ref{eq:system}), $x_{\mathrm{C}}(t)\in\mathbb{R}^{n_{\mathrm{C}}}$
is the internal state of the controller, $u_{\mathrm{C}}(t)\in\mathbb{R}^{m}$
is the control command transmitted from the controller, and $A_{\mathrm{C}}$,
$B_{\mathrm{C}}$, $C_{\mathrm{C}}$, $D_{\mathrm{C}}$ are matrices
that characterize the controller's dynamics. The closed-loop system
under this dynamic controller can be described by an equation similar
to (\ref{eq:closed-loop-system}). Specifically, by setting $\overline{x}(t)=[x^{\mathrm{T}}(t),x_{\mathrm{C}}^{\mathrm{T}}(t)]^{\mathrm{T}}$,
we have 
\begin{align}
\overline{x}(t+1)=\overline{A}\overline{x}(t)+(1-l(t))\overline{B}\overline{K}\overline{x}(t)+\overline{w}(t),\label{eq:overline-dynamics}
\end{align}
 where $\overline{w}(t)=[w_{\mathrm{P}}^{\mathrm{T}}(t)+(1-l(t))(Bw_{\mathrm{C}}(t))^{\mathrm{T}},0_{1\times n_{\mathrm{C}}}]^{\mathrm{T}}$
and
\begin{align*}
\overline{A} & =\left[\begin{array}{cc}
A & 0_{n\times n_{\mathrm{C}}}\\
B_{\mathrm{C}} & A_{\mathrm{C}}
\end{array}\right],\,\overline{B}=\left[\begin{array}{c}
B\\
0_{n_{\mathrm{C}}\times m}
\end{array}\right],\,\overline{K}=\left[\begin{array}{cc}
D_{\mathrm{C}} & C_{\mathrm{C}}\end{array}\right].
\end{align*}
 Output-feedback controllers can also be described similarly. We note
that dynamic controllers are shown to be advantageous in anytime-control
frameworks and soft real-time control systems \cite{greco2010design}.

Equation (\ref{eq:system}) represents the setting where the input
of the plant is set to $0$ whenever there is a transmission failure.
Similarly, we can consider the setting where the plant uses the previous
input value if there is a failure. In that case, the plant dynamics
is given by 
\begin{align*}
x(t+1) & =Ax(t)+Bv(t)l(t)+(1-l(t))Bu_{\mathrm{C}}(t)+w_{\mathrm{P}}(t),\\
v(t+1) & =v(t)l(t)+(1-l(t))u_{\mathrm{C}}(t),
\end{align*}
 where $v(t)\in\mathbb{R}^{m}$ represents the last control command
that was successfully transmitted from the controller. In the case
of the state-feedback controller (\ref{eq:control-input}), we can
let $\overline{x}(t)=[x^{\mathrm{T}}(t),v^{\mathrm{T}}(t)]^{\mathrm{T}}$
and describe the dynamics of the closed-loop system using (\ref{eq:overline-dynamics}),
where 
\begin{align*}
\overline{A} & =\left[\begin{array}{cc}
A & B\\
0_{m\times n} & I_{m}
\end{array}\right],\,\overline{B}=\left[\begin{array}{c}
B\\
I_{m}
\end{array}\right],\,\overline{K}=\left[\begin{array}{cc}
K & -I_{m}\end{array}\right],
\end{align*}
 with $I_{m}$ denoting the identity matrix in $\mathbb{R}^{m\times m}$
and $\overline{w}(t)=[w_{\mathrm{P}}^{\mathrm{T}}(t)+(1-l(t))(Bw_{\mathrm{C}}(t))^{\mathrm{T}},(1-l(t))w_{\mathrm{C}}^{\mathrm{T}}(t)]^{\mathrm{T}}$.
\hfill \hfill $\triangleleft$

\end{remark}

\subsection{Transmission Failure Model}

In our problem setting, the likelihood of a transmission failure depends
on the \emph{power }of the jamming interference. If the interference
power is large, then a transmission failure may likely occur. In particular,
with $\phi_{\mathrm{J}}(t)\in[0,\infty)$ denoting the interference
power at time $t$, the transmission failure indicator $l(t)$ in
(\ref{eq:system}) is given by 
\begin{align}
l(t) & \triangleq\mathbbm{1}[r(t)\leq p(\phi_{\mathrm{J}}(t))],\quad t\in\mathbb{N}_{0},\label{eq:ldef}
\end{align}
where $p\colon[0,\infty)\to[0,1]$ is a Borel-measurable, nondecreasing
function, and $r(0),r(1),\ldots$ are independent random variables
that are distributed uniformly in $[0,1]$. Furthermore $\{r(t)\in[0,1]\}_{t\in\mathbb{N}_{0}}$
and $\{\phi_{\mathrm{J}}(t)\in[0,\infty)\}_{t\in\mathbb{N}_{0}}$
are assumed to be mutually independent processes. Notice that for
a fixed scalar $\phi$, we represent by $p(\phi)$ the conditional
probability of a transmission failure given that the jamming interference
power is set to $\phi$. In particular, (\ref{eq:ldef}) implies 
\begin{align*}
\mathbb{P}[l(t)=1|\phi_{\mathrm{J}}(t)=\phi] & =\mathbb{P}[r(t)\leq p(\phi)|\phi_{\mathrm{J}}(t)=\phi]\\
 & =\mathbb{P}[r(t)\leq p(\phi)]=p(\phi).
\end{align*}
Observe that, if $\phi_{\mathrm{J}}(t)$ is large so that $p(\phi_{\mathrm{J}}(t))$
is close to $1$, then  it becomes more likely that $r(t)\leq p(\phi_{\mathrm{J}}(t))$,
and hence by (\ref{eq:ldef}), a transmission failure is likely to
occur. We note that the attacker controls the \emph{power level} $\phi_{\mathrm{J}}(t)$
of jamming signals, but \emph{not} the jamming signals themselves.

Note also that transmission failures at different times are \emph{conditionally
independent} given the interference powers at those times. Namely,
for every $t_{1}<t_{2}<\cdots<t_{k}$, $k\in\mathbb{N}$, 
\begin{align*}
 & \mathbb{P}[l(t_{1})=1,\ldots,l(t_{k})=1|\phi_{\mathrm{J}}(t_{1})=\phi_{1},\ldots,\phi_{\mathrm{J}}(t_{k})=\phi_{k}]\\
 & \quad=\prod_{i=1}^{k}\mathbb{P}[l(t_{i})=1|\phi_{\mathrm{J}}(t_{i})=\phi_{i}]=\prod_{i=1}^{k}p(\phi_{i}).
\end{align*}

The characterization in (\ref{eq:ldef}) enables us to describe security
properties of different wireless channel models, as illustrated below. 

\begin{example}[Outage probability] \label{ExampleOutage} The function
$p$ can be used for describing the\emph{ outage probability} in wireless
channels with fading. Outage occurs when the SINR at the receiver
side (the plant in this paper) goes below a threshold  due to fading
(see Section 14.2 in \cite{IFACproakis} and Section 12.2.3 in \cite{molisch2011book}).
Outage probability has been used in different problem settings that
involve jamming attacks \cite{zhu2013networked,sheikholeslami2014jamming}.
Here we present two examples. First, in the case of a Rayleigh-fading
channel considered in \cite{sheikholeslami2014jamming}, the outage
probability is given by 
\begin{align}
p(\phi_{\mathrm{J}}) & =1-\frac{e^{-\underline{\gamma}\sigma/(b_{2}\xi)}}{1+\underline{\gamma}(b_{1}\phi_{\mathrm{J}})/(b_{2}\xi)},\label{eq:p-outage}
\end{align}
 where $\xi\in(0,\infty)$ and $\sigma\in(0,\infty)$ are constants
associated respectively with the transmission power and the power
of the channel noise. The scalars $b_{1},b_{2}\in(0,\infty)$ depend
on the distances of the jamming attacker and the controller from the
plant. They are constant in our setup, since the geographical locations
of the jamming attacker, the controller, and the plant are fixed.
The scalar $\underline{\gamma}$ represents the SINR-threshold. As
the second example, we can also investigate the approximate outage
probability considered in \cite{zhu2013networked,sheikholeslami2014jamming,hu2015distributed}
by setting 
\begin{align}
p(\phi_{\mathrm{J}}) & =1-e^{-\underline{\gamma}/\gamma},\label{eq:p-outage-approximate}
\end{align}
 where $\gamma=\frac{b_{2}\xi}{b_{1}\phi_{\mathrm{J}}+\sigma}$ is
the SINR and $\underline{\gamma}$ is its threshold for outage. The
scenarios in \cite{zhu2013networked,hu2015distributed} involve moving
transmitters and interference sources. Our setup is closer to \cite{sheikholeslami2014jamming}
in that the jamming attacker is not mobile, but capable of changing
the power of emitted interference. \hfill $\triangleleft$\end{example}
\medskip

\begin{example} \label{ExamplePacketLength} Additive white Gaussian
noise channel models considered in \cite{IFACli2016sinr,IFACzhangattack}
can be represented by appropriately choosing $p$. For instance, a
special case of the model in \cite{IFACzhangattack} with fixed channel
gains can be represented with 
\begin{align}
p(\phi_{\mathrm{J}}) & =1-\big(1-Q(\sqrt{2\xi/(\phi_{\mathrm{J}}+\sigma)})\big)^{L},\label{eq:p-with-packet-length}
\end{align}
where $Q(y)\triangleq\frac{1}{\sqrt{2\pi}}\int_{y}^{\infty}e^{-\frac{s^{2}}{2}}\mathrm{d}s$,
$L\in\mathbb{N}$ denotes the length of packet being transmitted,
and the positive constants $\xi$ and $\sigma$ respectively denote
the transmission and the channel noise powers. \hfill $\triangleleft$
\end{example} \medskip

The transmission failure probability function $p$ in (\ref{eq:ldef})
plays an important role in the analysis presented in the next section.
In particular, if $p$ is a concave function and the average power
of jamming interference is upper-bounded by a scalar $\overline{\phi}_{\mathrm{J}}$
(as we explain later), then $p(\overline{\phi}_{\mathrm{J}})$ can
be used in the stability analysis as an upper bound on the long-run
average number of transmission failures (i.e., $\limsup_{t\to\infty}\frac{1}{t}\sum_{i=0}^{t-1}l(i)\leq p(\overline{\phi}_{\mathrm{J}})$,
almost surely). If $p$ is not concave, then a concave function that
upper-bounds $p$ can be used for the same purpose. To this end, in
this paper we use a continuous, nondecreasing, and concave function
$\hat{p}\colon[0,\infty)\to[0,1]$ such that 
\begin{align}
\hat{p}(\phi) & \geq p(\phi),\quad\phi\in[0,\infty).\label{eq:p_phat_ineq}
\end{align}
 Notice that such a function $\hat{p}$ always exists. The work \cite{ahmetifacwc2017}
discusses methods of finding tight concave upper-bounding functions
$\hat{p}$. Furthermore, in the case of the transmission failure probability
functions in (\ref{eq:p-outage}) and (\ref{eq:p-outage-approximate})
from Example~\ref{ExampleOutage}, it suffices to choose $\hat{p}$
same as $p$, since in both cases $p$ is continuous, nondecreasing,
and concave (by having a nonpositive second derivative).

\section{Analysis of Networked Stabilization}

\label{sec:Analysis-of-Networked}

In this section, we first provide a quick look at the stability of
networked control system (\ref{eq:closed-loop-system}) in the disturbance-free
case. Then we discuss how a strategic jamming attacker can take advantage
of the presence of disturbance to prevent stabilization. Finally,
we obtain conditions of stability under disturbance. 

\subsection{Stabilization in the Disturbance-Free Case}

A networked control system under jamming attacks but without disturbance
($w(t)=0$, $t\in\mathbb{N}_{0}$) was studied in \cite{ahmetifacwc2017}.
There, it was noted that emitting jamming interference signals is
a costly action due to its large energy requirements \cite{pelechrinis2011}.
The following assumption on the attacker's interference power was
considered in that work as a natural way to describe the energy constraints
of an attacker. 

\begin{assumption} \label{FirstAssumption} There exist scalars $\overline{\kappa}\geq0$,
$\overline{\phi}_{\mathrm{J}}\geq0$ such that 
\begin{align}
\mathbb{P}\big[\sum_{i=0}^{t-1}\phi_{\mathrm{J}}(i)\leq\overline{\kappa}+\overline{\phi}_{\mathrm{J}}t\big] & =1,\quad t\in\mathbb{N}.\label{eq:attack_assumption_one}
\end{align}
\end{assumption} \medskip

Here, the scalar $\overline{\kappa}$ models the attacker's initial
capabilities. Large $\overline{\kappa}$ values describe attackers
with large initial energy resources capable of setting $\phi_{\mathrm{J}}(t)$
to large values for a few initial time instants. On the other hand,
$\overline{\phi}_{\mathrm{J}}$ is an upper bound on the \emph{long-run}
average interference power (i.e., $\limsup_{k\to\infty}\frac{1}{k}\sum_{t=0}^{k-1}\phi_{\mathrm{J}}(t)\leq\overline{\phi}_{\mathrm{J}}$)
describing the overall attack strength. The scalar $\overline{\phi}_{\mathrm{J}}$
is typically strictly smaller than the maximum possible power of the
interference that can be physically emitted from the attacker. However,
by waiting sufficiently long without attacking, the attacker can preserve
energy and emit strong interference signals with power levels larger
than $\overline{\phi}_{\mathrm{J}}$ for certain durations while still
satisfying (\ref{eq:attack_assumption_one}). 

The analysis in \cite{ahmetifacwc2017} indicates that if the long-run
average power bound $\overline{\phi}_{\mathrm{J}}$ is sufficiently
small, then the closed-loop system (\ref{eq:closed-loop-system})
is asymptotically stable almost surely, implying $\mathbb{P}[\lim_{t\to\infty}\|x(t)\|_{2}=0]=1$.
A similar conclusion is drawn in Proposition A.1 
 in the Appendix
for the first-moment asymptotic stability, implying  $\lim_{t\to\infty}\mathbb{E}[\|x(t)\|_{2}]=0$. 

\subsection{Effects of Jamming Attacks on Systems Under Disturbance}

\label{subsec:Joint-Effect-of}

For certain systems that face disturbance,  jamming attacks can become
more dangerous. Even if the disturbance is very small and the attacker
has very limited resources, there still exist attack strategies that
can \emph{destabilize} the system while satisfying Assumption~\ref{FirstAssumption}
with very small $\overline{\phi}_{\mathrm{J}}$. We illustrate this
idea in the following example.

\begin{example} \label{ExampleStrategy} Consider a scalar networked
control system (\ref{eq:system}), (\ref{eq:control-input}) with
$x_{0}>0$, $A+BK\in[0,1)$, $A>1$, and a constant disturbance $w(t)=w^{*}>0$,
$t\in\mathbb{N}_{0}$, as a dynamic effect that is unrelated to jamming.
Suppose that the conditional probability $p$ of transmission failures
is a strictly increasing function (e.g., $p$ given by the outage
probability (\ref{eq:p-outage})). For this setup, an attacker can
wait sufficiently long and then attack for a duration with a sufficiently
large interference power level  so that the state norm grows to large
values but the average interference power does not go above $\overline{\phi}_{\mathrm{J}}$.
In particular, for any $\overline{\phi}_{\mathrm{J}}>0$, $x_{0}>0$,
$z>0$, and $\rho\in(0,1)$, the attack strategy
\begin{align}
\phi_{\mathrm{J}}(t) & \triangleq\begin{cases}
\phi_{\mathrm{J}}^{*}, & t\in\{\tau_{1},\ldots,\tau_{1}+\tau_{2}-1\},\\
0, & \mathrm{otherwise},
\end{cases}\label{eq:ExampleStrategyV}
\end{align}
 with $\phi_{\mathrm{J}}^{*}\triangleq p^{-1}(\rho^{\frac{1}{\tau_{2}}})+1$,
$\tau_{1}\triangleq\lfloor\frac{\max\{\phi_{\mathrm{J}}^{*}-\overline{\phi}_{\mathrm{J}},0\}\tau_{2}}{\overline{\phi}_{\mathrm{J}}}\rfloor+1$,
$\tau_{2}\triangleq\lfloor\max\{\log_{A}(z/w^{*}),0\}\rfloor+1$ guarantees
that Assumption~\ref{FirstAssumption} is satisfied and the state
exceeds the value $z$ with probability larger than $\rho$ at time
$\tau\triangleq\tau_{1}+\tau_{2}$, i.e., $\mathbb{P}[x(\tau)>z]>\rho.$ 

To show this, we consider
\begin{align*}
E(\tau_{1},\tau_{2}) & \triangleq\big\{\omega\in\Omega\colon l(t)=1,t\in\{\tau_{1},\ldots,\tau_{1}+\tau_{2}-1\}\big\}\in\mathcal{F},
\end{align*}
which represents the case where all packet transmissions during $t\in\{\tau_{1},\ldots,\tau_{1}+\tau_{2}-1\}$
fail. By (\ref{eq:ExampleStrategyV}), we have $\mathbb{P}[E(\tau_{1},\tau_{2})]=p^{\tau_{2}}(\phi_{\mathrm{J}}^{*})$.
Now, since $x_{0}>0$, $A>1$, and $w^{*}>0$, we obtain $x(t)\geq w^{*},$
$t\in\mathbb{N}$. Therefore, 
\begin{align*}
 & \mathbb{P}[x(\tau)>z]\geq\mathbb{P}[x(\tau)>z\,|\,E(\tau_{1},\tau_{2})]\mathbb{P}[E(\tau_{1},\tau_{2})]\\
 & \quad\geq\mathbb{P}[A^{\tau_{2}}x(\tau_{1})+\sum_{i=0}^{\tau_{2}-1}A^{i}w^{*}>z\,|\,E(\tau_{1},\tau_{2})]p^{\tau_{2}}(\phi_{\mathrm{J}}^{*})\\
 & \quad\geq\mathbb{P}[A^{\tau_{2}}w^{*}>z\,|\,E(\tau_{1},\tau_{2})]p^{\tau_{2}}(\phi_{\mathrm{J}}^{*})>1\cdot\rho^{\frac{\tau_{2}}{\tau_{2}}}=\rho.
\end{align*}
Furthermore, the attack strategy (\ref{eq:ExampleStrategyV}) satisfies
Assumption~\ref{FirstAssumption} with $\overline{\kappa}=0$, because
$\tau_{1}\geq\max\{\phi_{\mathrm{J}}^{*}-\overline{\phi}_{\mathrm{J}},0\}\tau_{2}/\overline{\phi}_{\mathrm{J}}\geq(\phi_{\mathrm{J}}^{*}-\overline{\phi}_{\mathrm{J}})\tau_{2}/\overline{\phi}_{\mathrm{J}}$,
and thus, $\sum_{i=0}^{\tau-1}\phi_{\mathrm{J}}(i)=\phi_{\mathrm{J}}^{*}\tau_{2}\leq\overline{\phi}_{\mathrm{J}}(\tau_{1}+\tau_{2})=\overline{\phi}_{\mathrm{J}}\tau$.
\hfill $\triangleleft$ \end{example}

The attack strategy (\ref{eq:ExampleStrategyV}) can make the state
grow arbitrarily large even if the interference power bound $\overline{\phi}_{\mathrm{J}}$
is small. This attack strategy is effective, because even if the attacker
initially waits for a long duration without attacking, the state never
reaches a small neighborhood of zero due to the disturbance. Hence,
after waiting for a while, the attacker can consecutively attack with
high interference powers to cause many transmission failures and make
the state norm grow to large values. This is further illustrated in
Section~\ref{sec:Illustrative-Numerical-Examples}. 

\begin{remark}Attack strategies similar to the one discussed in Example~\ref{ExampleStrategy}
may not always be able to leverage the existence of disturbance to
cause instability, if the disturbance only affects a stable mode of
the system that does not get influenced by jamming-related transmission
failures. For instance, consider (\ref{eq:system}), (\ref{eq:control-input})
where
\begin{align}
 & A=\left[\begin{array}{cc}
0.5 & 0\\
0 & 2
\end{array}\right],\,B=\left[\begin{array}{c}
0\\
1
\end{array}\right],\,K=\left[\begin{array}{cc}
0 & -1.5\end{array}\right].\label{eq:abkstableunstable}
\end{align}
In this case, consider a disturbance process $w(t)=[w_{1}(t),0]^{\mathrm{T}}$.
Here, the disturbance only affects the first state and the jamming
attacks only affect the second state. The first state remains bounded
under bounded disturbances regardless of jamming. On the other hand,
sufficiently frequent transmission failures due to jamming attacks
can cause the second state to diverge. However, differently from Example~\ref{ExampleStrategy},
the attacker needs to spend considerably more resources to cause instability,
because disturbance and jamming affect different parts of the dynamics.
In the following sections, we provide conditions of stabilization
by considering the worst-case scenarios as in Example~\ref{ExampleStrategy}.
Direct application of such conditions can be conservative for systems
similar to (\ref{eq:abkstableunstable}). However, in some cases,
conservativeness can be reduced. For instance, if a system has components
that are influenced by both the jamming and the disturbance, partitioning
the system to apply the results to only those components can help
reduce conservativeness. \hfill $\triangleleft$ \end{remark} 

\subsection{Jamming Interference and Bounded Disturbance}

To ensure stability under both disturbance and jamming, the attacks
need to be restricted in a way that high jamming interference powers
at consecutive times are not allowed. To this end, we consider the
following assumption. 

\begin{assumption} \label{SecondAssumption} There exist scalars
$\hat{\kappa}\geq0$ and $\hat{\phi}_{\mathrm{J}}\geq0$ such that
\begin{align}
\mathbb{P}\big[\sum_{i=t_{1}}^{t_{2}-1}\phi_{\mathrm{J}}(i)\leq\hat{\kappa}+\hat{\phi}_{\mathrm{J}}(t_{2}-t_{1})\big] & =1,\label{eq:attack_assumption_two}
\end{align}
 for all $t_{1},t_{2}\in\mathbb{N}_{0}$ with $t_{1}<t_{2}$. \end{assumption}

Notice that (\ref{eq:attack_assumption_two}) implies (\ref{eq:attack_assumption_one})
(with $\overline{\kappa}=\hat{\kappa}$ and $\overline{\phi}_{\mathrm{J}}=\hat{\phi}_{\mathrm{J}}$),
but the converse is not true. Assumption~\ref{SecondAssumption}
is thus more restrictive than Assumption~\ref{FirstAssumption}.
In particular, under Assumption~\ref{SecondAssumption}, the attacker
can attack with a jamming interference power $\phi_{\mathrm{J}}^{*}>\hat{\phi}_{\mathrm{J}}$
consecutively for at most $\lfloor\hat{\kappa}/(\phi_{\mathrm{J}}^{*}-\hat{\phi}_{\mathrm{J}})\rfloor$
time steps; hence, the destabilizing attacks discussed in Example~\ref{ExampleStrategy}
are avoided. Furthermore, setting $\lfloor\hat{\kappa}/(\phi_{\mathrm{J}}^{\max}-\hat{\phi}_{\mathrm{J}})\rfloor=1$
and solving for $\phi_{\mathrm{J}}^{\max}$ provide us a hard constraint
($\phi_{\mathrm{J}}(t)\leq\phi_{\mathrm{J}}^{\max}$ for every $t\in\mathbb{N}_{0}$)
on the interference power level. The value $\phi_{\mathrm{J}}^{\max}$
is related to the physical limits of interference generation in wireless
jamming units. 

Assumption~\ref{SecondAssumption} is related to other characterizations
of attacks. In particular, in the continuous-time deterministic DoS
attack characterization of \cite{de2015inputtran}, the number of
attacks in a given time frame as well as the total duration of those
attacks are bounded by certain ratios of the length of that time frame.
Under that characterization, the maximum possible length of a continuous
attack duration is bounded, which enables analysis of input-to-state
stability under disturbance. The restriction on jamming through Assumption~\ref{SecondAssumption}
is similar, since long consecutive emissions of high-powered  interference
signals are not allowed. We note, however, that Assumption~\ref{SecondAssumption}
does allow the case where the channel is attacked at all times if
the attacker's interference power for certain times is small. 

In this section, we investigate the networked control system (\ref{eq:closed-loop-system})
under bounded disturbances. The analysis is then extended in Section~\ref{subsec:Stabilization-Under-Jamming}
to the case where the disturbance has finite second moments but its
norm may not be bounded by a fixed scalar. 

In this paper, we consider scenarios where the norm of the disturbance
does not approach zero, and hence the state or its moments may not
converge to the origin. Instead of exploring asymptotic stability,
our goal is to obtain conditions for the first moment of the state
to stay bounded. To this end, let $\hat{A}(t)\triangleq l(t)A+(1-l(t))(A+BK)$,
$t\in\mathbb{N}_{0}$, and moreover, for every $t_{1},t_{2}\in\mathbb{N}_{0}$
with $t_{1}\leq t_{2},$ let 
\begin{align*}
F(t_{2},t_{1}) & \triangleq\begin{cases}
\hat{A}(t_{2}), & t_{1}=t_{2},\\
\hat{A}(t_{2})\cdots\hat{A}(t_{1}), & t_{1}<t_{2}.
\end{cases}
\end{align*}
For the closed-loop system (\ref{eq:closed-loop-system}), we have
$x(t)=F(t-1,0)x_{0}+\sum_{j=0}^{t-2}F(t-1,j+1)w(j)+w(t-1)$, for $t\in\mathbb{N}$.
Therefore, for any induced norm $\|\cdot\|$, it follows from the
triangle inequality and the submultiplicativity property that $\|x(t)\|\leq\Big(\prod_{i=0}^{t-1}\|\hat{A}(i)\|\Big)\|x_{0}\|+\sum_{j=0}^{t-2}\Big(\prod_{i=j+1}^{t-1}\|\hat{A}(i)\|\Big)\|w(j)\|+\|w(t-1)\|.$
Here, we have $\|\hat{A}(i)\|=l(i)\|A\|+(1-l(i))\|A+BK\|$, $i\in\mathbb{N}_{0}$.
Hence, by letting 
\begin{align}
\zeta_{1} & \triangleq\|A\|-\|A+BK\|,\quad\zeta_{0}\triangleq\|A+BK\|,\label{eq:zeta-def}
\end{align}
 we obtain for $t\in\mathbb{N}$, $\|x(t)\|\leq\Big(\prod_{i=0}^{t-1}(\zeta_{1}l(i)+\zeta_{0})\Big)\|x_{0}\|+\sum_{j=0}^{t-2}\Big(\prod_{i=j+1}^{t-1}(\zeta_{1}l(i)+\zeta_{0})\Big)\|w(j)\|+\|w(t-1)\|$.
By using this inequality, we can also obtain an upper bound of the
Euclidean norm of the state. Specifically, by Corollary~5.4.5 of
\cite{hornmatrixanalysis}, there exist $c_{1}>0$ and $c_{2}>c_{1}$
such that 
\begin{align}
c_{1}\|y\| & \leq\|y\|_{2}\leq c_{2}\|y\|,\quad y\in\mathbb{R}^{n}.\label{eq:norm-equivalence}
\end{align}
Therefore, we have
\begin{align}
\|x(t)\|_{2} & \leq\frac{c_{2}}{c_{1}}\Bigg(\Big(\prod_{i=0}^{t-1}(\zeta_{1}l(i)+\zeta_{0})\Big)\|x_{0}\|_{2}\nonumber \\
 & \quad+\sum_{j=0}^{t-2}\Big(\prod_{i=j+1}^{t-1}(\zeta_{1}l(i)+\zeta_{0})\Big)\|w(j)\|_{2}+\|w(t-1)\|_{2}\Bigg).\label{eq:two-norm-ineq}
\end{align}
Notice here that the particular values of $c_{1}$ and $c_{2}$ depend
on the choice of the vector norm that induces the matrix norm $\|\cdot\|$. 

We use (\ref{eq:two-norm-ineq}) to provide bounds on the first moment
$\mathbb{E}[\|x(t)\|_{2}]$. First, in the following result, we consider
the case where the disturbance is \emph{bounded} and the jamming attacks
satisfy Assumption~\ref{SecondAssumption}. 

\begin{theorem} \label{TheoremBoundedDisturbance} Consider the closed-loop
networked control system (\ref{eq:closed-loop-system}). Suppose that
the attacker's interference power process $\{\phi_{\mathrm{J}}(t)\in[0,\infty)\}_{t\in\mathbb{N}_{0}}$
satisfies Assumption~\ref{SecondAssumption}. Furthermore, suppose
that there exists $\overline{w}\geq0$ such that 
\begin{align}
\mathbb{P}[\|w(t)\|_{2}\leq\overline{w}] & =1,\quad t\in\mathbb{N}.\label{eq:almost_sure_noise_bound}
\end{align}
If 
\begin{align}
 & (1-\hat{p}(\hat{\phi}_{\mathrm{J}}))\|A+BK\|+\hat{p}(\hat{\phi}_{\mathrm{J}})\|A\|<1,\label{eq:abkcond-vhat}
\end{align}
 then there exist $\hat{\mu}\geq0$, $\hat{\theta}\in(0,1)$, and
$\hat{d}\geq0$ such that 
\begin{align}
\mathbb{E}[\|x(t)\|_{2}] & \leq\hat{\mu}\hat{\theta}^{t}\|x_{0}\|_{2}+\hat{d}\overline{w},\quad t\in\mathbb{N}.\label{eq:noise_to_state_one}
\end{align}

\end{theorem}

\medskip

The proof of Theorem~\ref{TheoremBoundedDisturbance} is given later
in the paper. Theorem~\ref{TheoremBoundedDisturbance} shows that
if jamming attacks satisfy Assumption~\ref{SecondAssumption} with
a sufficiently small $\hat{\phi}_{\mathrm{J}}$ such that (\ref{eq:abkcond-vhat})
holds, then the first moment of the state stays bounded. Furthermore,
the upper bound given in (\ref{eq:noise_to_state_one}) is geometrically
decreasing towards the constant $\hat{d}\overline{w}$, where $\overline{w}$
is an upper bound on the Euclidean norm of disturbance $w(t)$. 

 Notice that the condition (\ref{eq:abkcond-vhat})
of Theorem~\ref{TheoremBoundedDisturbance} and the condition (\ref{eq:abkcond})
in the disturbance-free case in Proposition~\ref{PropositionMomentConvergence}
are in the same form, but use different scalars $\hat{\phi}_{\mathrm{J}}$
and $\overline{\phi}_{\mathrm{J}}$ due to the difference of the jamming
interference characterizations in Assumptions~\ref{FirstAssumption}
and \ref{SecondAssumption}. We remark that for attacks that satisfy
both assumptions, we have $\overline{\phi}_{\mathrm{J}}\leq\hat{\phi}_{\mathrm{J}}$.

As we establish later in the proof of Theorem~\ref{TheoremBoundedDisturbance},
the first-moment upper bound in (\ref{eq:noise_to_state_one}) depends
on parameters $\hat{\kappa}$ and $\hat{\phi}_{\mathrm{J}}$. In particular,
the values of $\hat{\mu}$ and $\hat{d}$ are large when $\hat{k}$
and $\hat{\phi}_{\mathrm{J}}$ take large values. Moreover, the scalar
$\hat{\theta}$ is directly related to the term $(1-\hat{p}(\hat{\phi}_{\mathrm{J}}))\|A+BK\|+\hat{p}(\hat{\phi}_{\mathrm{J}})\|A\|$
on the left-hand side of (\ref{eq:abkcond-vhat}). If this term is
close to zero, then $\hat{\theta}$ is close to zero, which indicates
faster convergence of the bound in (\ref{eq:noise_to_state_one})
towards the constant $\hat{d}\overline{w}$. We note that $(1-\hat{p}(\hat{\phi}_{\mathrm{J}}))\|A+BK\|+\hat{p}(\hat{\phi}_{\mathrm{J}})\|A\|$
represents the behavior of the overall networked control system and
it is composed of the convex combination of the terms $\|A+BK\|$
and $\|A\|$ weighted respectively with the lower bound $(1-\hat{p}(\hat{\phi}_{\mathrm{J}}))$
of the long-term ratio of successful transmissions and the upper bound
$\hat{p}(\hat{\phi}_{\mathrm{J}})$ of the long-term ratio of failed
transmissions.

Our analytical approach differs from the more classical approaches
used when the transmission failure indicator process $\{l(t)\}_{t\in\mathbb{N}_{0}}$
is a Bernoulli process or a Markov chain. In those cases, stability
analysis can rely on the probability of failures $\mathbb{P}[l(t)=1]$
and conditional failure probabilities $\mathbb{P}[l(t)=q|l(t-1)=r]$,
$q,r\in\{0,1\}$, (see \cite{quevedo2011packetized,costa2004discrete}).
In our case, precise information of such probability terms is not
available due to the uncertainty in the generation of attacks. Specifically,
the interference power $\phi_{\mathrm{J}}(t)$ at a given time $t$
is part of attacker's strategy and cannot be known with certainty.
As a result, the transmission failure probability at that time is
also uncertain and cannot be used in the analysis. Note, however,
that Bernoulli-type packet losses are a special case where $\phi_{\mathrm{J}}$
is a constant function.  In this paper, we are interested in the cases
where the interference power level $\phi_{\mathrm{J}}$ is time-varying
and the attacker designs its progression so as to leverage the disturbance
to cause instability as in Example~\ref{ExampleStrategy}. 

A crucial role in our analysis is played by the following lemma, where
we investigate the products of affine functions that involve the transmission
failure indicator $l(\cdot)$ and obtain some upper bounds for their
expected values. As shown later in the proof of Theorem~\ref{TheoremBoundedDisturbance},
such upper bounds allow us to conduct stability analysis without relying
on transmission failure probabilities for each time step. In the derivation
of these bounds, an essential step is to exploit the concavity of
the upper-bounding function $\hat{p}$ given in (\ref{eq:p_phat_ineq}). 

\begin{lemma} \label{LemmaDOne} Suppose that the attacker's interference
power process  $\{\phi_{\mathrm{J}}(t)\in[0,\infty)\}_{t\in\mathbb{N}_{0}}$
satisfies Assumption~\ref{SecondAssumption}. Then for every $\alpha_{1}\geq0$,
$\alpha_{0}\geq0$ that satisfy
\begin{align}
\alpha_{1}\hat{p}(\hat{\phi}_{\mathrm{J}})+\alpha_{0} & <1,\label{eq:alpha_ineq}
\end{align}
there exist scalars $\mu\geq0$ and $\theta\in(0,1)$ such that 
\begin{align}
\mathbb{E}\Big[\prod_{i=t_{1}}^{t_{2}-1}(\alpha_{1}l(i)+\alpha_{0})\Big] & \leq\mu\theta^{(t_{2}-t_{1})},\label{eq:geometric-lemma-result}
\end{align}
for $t_{1},t_{2}\in\mathbb{N}_{0}$ with $t_{1}<t_{2}$. \end{lemma}

\begin{proof} For the case where $\alpha_{1}+\alpha_{0}=0$, (\ref{eq:geometric-lemma-result})
holds for any $\mu\geq0$ and $\theta\in(0,1)$. In the following,
we consider the case where $\alpha_{1}+\alpha_{0}>0$. First, by Lemma~2.1
of \cite{ahmetifacwc2017}, 
\begin{equation}
\mathbb{E}\Big[\prod_{i=t_{1}}^{t_{2}-1}(\alpha_{1}l(i)+\alpha_{0})\Big]=\mathbb{E}\Big[\prod_{i=t_{1}}^{t_{2}-1}(\alpha_{1}p(\phi_{\mathrm{J}}(i))+\alpha_{0})\Big].\label{eq:first-eq}
\end{equation}
Next, by (\ref{eq:first-eq}), $\alpha_{1}\geq0$, and $p(\phi)\leq\hat{p}(\phi)$,
$\phi\in[0,\infty)$, we get 
\begin{align}
\mathbb{E}\Big[\prod_{i=t_{1}}^{t_{2}-1}(\alpha_{1}l(i)+\alpha_{0})\Big] & \leq\mathbb{E}\Big[\prod_{i=t_{1}}^{t_{2}-1}h(\phi_{\mathrm{J}}(i))\Big],\label{eq:expxhineq-1-1}
\end{align}
where $h(\phi)\triangleq\alpha_{1}\hat{p}(\phi)+\alpha_{0}$. We note
that $h(\cdot)$ is nondecreasing, concave, and continuous, as $\hat{p}(\cdot)$
also has such properties and $\alpha_{1}\geq0$. To obtain a bound
for $\mathbb{E}[\prod_{i=t_{1}}^{t_{2}-1}h(\phi_{\mathrm{J}}(i))]$
in (\ref{eq:expxhineq-1-1}), we first show 
\begin{align}
\prod_{i=t_{1}}^{t_{2}-1}h(\phi_{\mathrm{J}}(i)) & \leq h^{(t_{2}-t_{1})}\Big(\frac{1}{t_{2}-t_{1}}\sum_{i=t_{1}}^{t_{2}-1}\phi_{\mathrm{J}}(i)\Big).\label{eq:powerbound}
\end{align}
We note that (\ref{eq:powerbound}) holds if $h(\phi_{\mathrm{J}}(i))=0$
for some $i\in\{t_{1},\ldots,t_{2}-1\}$. Now, consider the case where
$h(\phi_{\mathrm{J}}(i))>0$ for all $i\in\{t_{1},\ldots,t_{2}-1\}$.
For this case, we have 
\begin{align}
\ln\prod_{i=t_{1}}^{t_{2}-1}h(\phi_{\mathrm{J}}(i)) & =(t_{2}-t_{1})\Big(\frac{1}{t_{2}-t_{1}}\sum_{i=t_{1}}^{t_{2}-1}\ln h(\phi_{\mathrm{J}}(i))\Big).\label{eq:firstln}
\end{align}
 Here, $\ln h(\cdot)$ is concave, since it is the composition of
a nondecreasing concave function $\ln(\cdot)$ and a concave function
$h(\cdot)$ (see Proposition 2.16 in \cite{avriel2010generalized}
and Section 3.2.4 in \cite{IFACboyd2004convex}). Thus, by (\ref{eq:firstln}),
\begin{align}
\ln\prod_{i=t_{1}}^{t_{2}-1}h(\phi_{\mathrm{J}}(i)) & \leq(t_{2}-t_{1})\ln h\Big(\frac{1}{t_{2}-t_{1}}\sum_{i=t_{1}}^{t_{2}-1}\phi_{\mathrm{J}}(i)\Big),\label{eq:lnineq2}
\end{align}
 which implies (\ref{eq:powerbound}). The interference power process
$\phi_{\mathrm{J}}(\cdot)$ satisfies (\ref{eq:attack_assumption_two})
in Assumption~\ref{SecondAssumption}, and hence, $\frac{1}{t_{2}-t_{1}}\sum_{i=t_{1}}^{t_{2}-1}\phi_{\mathrm{J}}(i)\leq\frac{\hat{\kappa}}{t_{2}-t_{1}}+\hat{\phi}_{\mathrm{J}}$,
almost surely. Thus, noting that $h(\cdot)$ is a nondecreasing function,
by (\ref{eq:powerbound}), we obtain $\prod_{i=t_{1}}^{t_{2}-1}h(\phi_{\mathrm{J}}(i))\leq h^{t_{2}-t_{1}}(\frac{\hat{\kappa}}{t_{2}-t_{1}}+\hat{\phi}_{\mathrm{J}})$,
almost surely. Consequently, we have 
\begin{align}
\mathbb{E}\Big[\prod_{i=t_{1}}^{t_{2}-1}h(\phi_{\mathrm{J}}(i))\Big] & \leq h^{t_{2}-t_{1}}\Big(\frac{\hat{\kappa}}{t_{2}-t_{1}}+\hat{\phi}_{\mathrm{J}}\Big).\label{eq:hfineq-1-1}
\end{align}
Now, by (\ref{eq:alpha_ineq}), we get $h(\hat{\phi}_{\mathrm{J}})<1$.
Therefore, by the continuity of $h(\cdot)$, there exists $\delta>0$
such that $h(\delta+\hat{\phi}_{\mathrm{J}})<1$. As a result, for
sufficiently large values of $t_{2}-t_{1}$, we have $h(\frac{\hat{\kappa}}{t_{2}-t_{1}}+\hat{\phi}_{\mathrm{J}})<1$. 

Let $T^{*}$ be a positive integer such that $h(\frac{\hat{\kappa}}{T^{*}}+\hat{\phi}_{\mathrm{J}})<1$
and let 
\begin{align}
\theta & \triangleq h\Big(\frac{\hat{\kappa}}{T^{*}}+\hat{\phi}_{\mathrm{J}}\Big).\label{eq:ThetaDef}
\end{align}
It follows from (\ref{eq:hfineq-1-1}) that 
\begin{align}
\mathbb{E}\Big[\prod_{i=t_{1}}^{t_{2}-1}h(\phi_{\mathrm{J}}(i))\Big] & \leq\theta^{t_{2}-t_{1}},\label{eq:hjt_part1}
\end{align}
 for all $t_{1},t_{2}\in\mathbb{N}_{0}$ such that $t_{2}-t_{1}\geq T^{*}$.
If $T^{*}=1$, then (\ref{eq:geometric-lemma-result}) holds, by (\ref{eq:hjt_part1}).
If, on the other hand, $T^{*}>1$, then by using $h(\phi_{\mathrm{J}}(t))\leq\alpha_{1}+\alpha_{0}$,
$t\in\mathbb{N}_{0}$, we obtain 
\begin{align}
\mathbb{E}\Big[\prod_{i=t_{1}}^{t_{2}-1}h(\phi_{\mathrm{J}}(i))\Big] & \leq(\alpha_{1}+\alpha_{0})^{t_{2}-t_{1}}\leq(\alpha_{1}+\alpha_{0})^{T^{*}-1},\label{eq:hjt_part2}
\end{align}
 for all $t_{1},t_{2}\in\mathbb{N}_{0}$ such that $0<t_{2}-t_{1}<T^{*}$.
Letting
\begin{align}
\mu & \triangleq(\alpha_{1}+\alpha_{0})^{T^{*}-1}\theta^{-(T^{*}-1)},\label{eq:MuDef}
\end{align}
we obtain (\ref{eq:geometric-lemma-result}), by (\ref{eq:hjt_part1})
and (\ref{eq:hjt_part2}). \end{proof}

Lemma~\ref{LemmaDOne} shows that under Assumption~\ref{SecondAssumption},
the expectation term $\mathbb{E}[\prod_{i=t_{1}}^{t_{2}-1}(\alpha_{1}l(i)+\alpha_{0})]$
with $\alpha_{1}\geq0$, $\alpha_{0}\geq0$ satisfying (\ref{eq:alpha_ineq}),
converges to zero at a geometric rate. By using this lemma, we obtain
the following result. 

\begin{lemma} \label{LemmaSum} Suppose that the attacker's interference
power process $\{\phi_{\mathrm{J}}(t)\in[0,\infty)\}_{t\in\mathbb{N}_{0}}$
satisfies Assumption~\ref{SecondAssumption}. Then for every  $\alpha_{1}\geq0$,
$\alpha_{0}\geq0$ satisfying (\ref{eq:alpha_ineq}), there exists
a scalar $d\geq0$ such that 
\begin{align}
\sum_{j=0}^{t-2}\mathbb{E}\Big[\prod_{i=j+1}^{t-1}(\alpha_{1}l(i)+\alpha_{0})\Big] & \leq d,\quad t\in\{2,3,\ldots\}.\label{eq:sum_result}
\end{align}
\end{lemma} \medskip

\begin{proof}Since (\ref{eq:alpha_ineq}) holds, it follows from
Lemma~\ref{LemmaDOne} that $\mathbb{E}[\prod_{i=j+1}^{t-1}(\alpha_{1}l(i)+\alpha_{0})]\leq\mu\theta^{(t-j-1)}$,
where $\mu\geq0$ and $\theta\in(0,1)$ are scalars that depend on
$\alpha_{1}$ and $\alpha_{0}$. Letting 
\begin{align}
d & \triangleq\mu/(1-\theta),\label{eq:DDef}
\end{align}
we obtain$\sum_{j=0}^{t-2}\mathbb{E}\Big[\prod_{i=j+1}^{t-1}(\alpha_{1}l(i)+\alpha_{0})\Big]\leq\sum_{j=0}^{t-2}\mu\theta^{(t-j-1)}=\mu\sum_{i=1}^{t-1}\theta^{i}\leq\mu\sum_{i=0}^{\infty}\theta^{i}=d$,
which completes the proof. \end{proof}

In Lemmas~\ref{LemmaDOne} and \ref{LemmaSum}, we obtained the upper-bounding
inequalities (\ref{eq:geometric-lemma-result}) and (\ref{eq:sum_result})
concerning the transmission failure indicator process $\{l(t)\in\{0,1\}\}_{t\in\mathbb{N}_{0}}$.
In our proof of Theorem~\ref{TheoremBoundedDisturbance} given below,
we utilize these inequalities. 

\emph{Proof of Theorem \ref{TheoremBoundedDisturbance}}: By (\ref{eq:two-norm-ineq})
and (\ref{eq:almost_sure_noise_bound}), $\|x(t)\|_{2}\leq\frac{c_{2}}{c_{1}}\left(\prod_{i=0}^{t-1}(\zeta_{1}l(i)+\zeta_{0})\right)\|x_{0}\|_{2}+\frac{c_{2}}{c_{1}}\Big(\sum_{j=0}^{t-2}\big(\prod_{i=j+1}^{t-1}(\zeta_{1}l(i)+\zeta_{0})\big)+1\Big)\overline{w},$
almost surely, and hence, for $t\in\mathbb{N}$, 
\begin{align}
\mathbb{E}[\|x(t)\|_{2}] & \leq\frac{c_{2}}{c_{1}}\mathbb{E}\Big[\prod_{i=0}^{t-1}(\zeta_{1}l(i)+\zeta_{0})\Big]\|x_{0}\|_{2}\nonumber \\
 & \quad+\frac{c_{2}}{c_{1}}\Big(\sum_{j=0}^{t-2}\mathbb{E}\Big[\prod_{i=j+1}^{t-1}(\zeta_{1}l(i)+\zeta_{0})\Big]+1\Big)\overline{w}.\label{eq:expectation_bounded_noise_case}
\end{align}

Next, we apply Lemmas~\ref{LemmaDOne} and \ref{LemmaSum} to obtain
upper bounds for the expectation terms on the right-hand side of (\ref{eq:expectation_bounded_noise_case}).
First, since $\|A\|>1$, (\ref{eq:abkcond-vhat}) implies $\|A+BK\|\in[0,1)$.
Thus, we have $\zeta_{1}>0$ and $\zeta_{0}\in[0,1)$. By letting
$\alpha_{1}=\zeta_{1}$ and $\alpha_{0}=\zeta_{0}$, (\ref{eq:abkcond-vhat})
implies (\ref{eq:alpha_ineq}). Therefore, by Lemmas~\ref{LemmaDOne}
and \ref{LemmaSum}, we have $\mathbb{E}\big[\prod_{i=0}^{t-1}(\zeta_{1}l(i)+\zeta_{0})\big]\leq\mu\theta^{t}$
and $\sum_{j=0}^{t-2}\mathbb{E}\big[\prod_{i=j+1}^{t-1}(\zeta_{1}l(i)+\zeta_{0})\big]\leq d$,
where $\mu\geq0$, $\theta\in(0,1)$, and $d\geq0$ are scalars that
depend on $\zeta_{1}$ and $\zeta_{0}$. Hence, with
\begin{align}
 & \hat{\theta}\triangleq\theta,\quad\hat{\mu}\triangleq\mu c_{2}/c_{1},\quad\hat{d}\triangleq(d+1)c_{2}/c_{1},\label{eq:thetamudhat}
\end{align}
 the inequality (\ref{eq:noise_to_state_one}) follows from (\ref{eq:expectation_bounded_noise_case}).
\hfill $\blacksquare$

\begin{remark} \label{RemarkFirstHatValues} By using (\ref{eq:thetamudhat})
together with (\ref{eq:ThetaDef}), (\ref{eq:MuDef}), and (\ref{eq:DDef}),
we can obtain the values of $\hat{\theta}$, $\hat{\mu}$, and $\hat{d}$
as 
\begin{align}
 & \hat{\theta}=(1-\hat{p}(\frac{\hat{\kappa}}{T^{*}}+\hat{\phi}_{\mathrm{J}}))\|A+BK\|+\hat{p}(\frac{\hat{\kappa}}{T^{*}}+\hat{\phi}_{\mathrm{J}})\|A\|,\label{eq:ThetaHatValue}\\
 & \hat{\mu}=(c_{2}/c_{1})\|A\|{}^{T^{*}-1}\hat{\theta}^{-(T^{*}-1)},\label{eq:MuHatValue}\\
 & \hat{d}=(c_{2}/c_{1})(\|A\|{}^{T^{*}-1}\hat{\theta}^{-(T^{*}-1)}/(1-\hat{\theta})+1),\label{eq:DHatValue}
\end{align}
 where $T^{*}$ is a positive integer that satisfies 
\begin{align*}
 & (1-\hat{p}(\frac{\hat{\kappa}}{T^{*}}+\hat{\phi}_{\mathrm{J}}))\|A+BK\|+\hat{p}(\frac{\hat{\kappa}}{T^{*}}+\hat{\phi}_{\mathrm{J}})\|A\|<1,
\end{align*}
and $c_{1},c_{2}>0$ satisfy (\ref{eq:norm-equivalence}). Note that
such $T^{*}$ always exists. An attacker with large resources can
cause the state norm to grow large. This is also indicated in the
upper bound for the first moment in (\ref{eq:noise_to_state_one}).
If $\hat{\kappa}$ in Assumption~\ref{SecondAssumption} is large,
then $T^{*}$ is large, which makes $\hat{\mu}$ large, as $\hat{\mu}$
is an increasing function of $T^{*}$. Further, since 
\begin{align*}
 & \frac{c_{2}}{c_{1}}\Big(\frac{\hat{\mu}}{1-(1-\hat{p}(\hat{\phi}_{\mathrm{J}}))\|A+BK\|-\hat{p}(\hat{\phi}_{\mathrm{J}})\|A\|}+1\Big)\leq\hat{d},
\end{align*}
 we observe that $\hat{d}$ is large for large values of $\hat{\mu}$
and $\hat{\phi}_{\mathrm{J}}$. On the other hand, for large values
of $T^{*}$, $\hat{\theta}$ is close to $(1-\hat{p}(\hat{\phi}_{\mathrm{J}}))\|A+BK\|+\hat{p}(\hat{\phi}_{\mathrm{J}})\|A\|$.
If the upper bound $\hat{\phi}_{\mathrm{J}}$ of average interference
powers is large, then $\hat{\mu}\hat{\theta}^{t}$ in (\ref{eq:noise_to_state_one})
converges slowly, since $\hat{\theta}$ is close to $1$. \end{remark}

\subsection{Jamming and Disturbance with Finite Second Moment}

\label{subsec:Stabilization-Under-Jamming}

In Theorem~\ref{TheoremBoundedDisturbance}, we explored the case
where the disturbance norm is bounded. Next, we investigate scenarios
where the disturbance may not be bounded. We obtain a relation between
the state and the disturbance similar to those used for noise-to-state
stability analysis of stochastic systems (e.g., \cite{nunez2014,zhang2016noise}).
Specifically, in the next result, we provide an upper bound for the
first moment of the state by using the second moment of the disturbance.

\begin{theorem} \label{TheoremDependentDisturbance} Consider the
closed-loop networked control system (\ref{eq:closed-loop-system}).
Suppose that the attacker's interference power process $\{\phi_{\mathrm{J}}(t)\in[0,\infty)\}_{t\in\mathbb{N}_{0}}$
satisfies Assumption~\ref{SecondAssumption}. Furthermore, suppose
$\mathbb{E}[\|w(t)\|_{2}^{2}]<\infty$, $t\in\mathbb{N}_{0}$. If
\begin{align}
 & (1-\hat{p}(\hat{\phi}_{\mathrm{J}}))\|A+BK\|^{2}+\hat{p}(\hat{\phi}_{\mathrm{J}})\|A\|^{2}<1,\label{eq:abkcond-difficult}
\end{align}
 then there exist $\hat{\mu},\hat{f}\geq0$, and $\hat{\theta}\in(0,1)$
 such that for $t\in\mathbb{N}$,
\begin{align}
\mathbb{E}[\|x(t)\|_{2}]\leq\hat{\mu}\hat{\theta}^{t}\|x_{0}\|_{2}+\hat{f}\max_{i\in\{0,\ldots,t-1\}}(\mathbb{E}[\|w(i)\|_{2}^{2}])^{\frac{1}{2}}.\label{eq:noise_to_state_three}
\end{align}
\end{theorem} \medskip

The proof of this result relies on the following lemma. 

\begin{lemma} \label{LemmaSumTwo} Suppose that the attacker's interference
power process $\{\phi_{\mathrm{J}}(t)\in[0,\infty)\}_{t\in\mathbb{N}_{0}}$
satisfies Assumption~\ref{SecondAssumption}. Then for every  $\gamma_{1}\geq0$,
$\gamma_{0}\geq0$ that satisfy 
\begin{align}
(\gamma_{1}^{2}+2\gamma_{1}\gamma_{0})\hat{p}(\hat{\phi}_{\mathrm{J}})+\gamma_{0}^{2} & <1,\label{eq:gamma_cond}
\end{align}
 there exists a scalar $f\geq0$ such that 
\begin{align}
\sum_{j=0}^{t-2}\Big(\mathbb{E}\Big[\big(\prod_{i=j+1}^{t-1}(\gamma_{1}l(i)+\gamma_{0})\big)^{2}\Big]\Big)^{\frac{1}{2}} & \leq f,\quad t\in\{2,3,\ldots\}.\label{eq:l_d1_ineq-1}
\end{align}
\end{lemma} \medskip

\begin{proof} For every $j\in\{0,1,\ldots,t-2\}$, $t\in\{2,3,\ldots\}$,
we have $(\prod_{i=j+1}^{t-1}(\gamma_{1}l(i)+\gamma_{0})){}^{2}=\prod_{i=j+1}^{t-1}(\gamma_{1}l(i)+\gamma_{0})^{2}=\prod_{i=j+1}^{t-1}(\gamma_{1}^{2}l^{2}(i)+2\gamma_{1}\gamma_{0}l(i)+\gamma_{0}^{2}).$
Let $\alpha_{1}\triangleq\gamma_{1}^{2}+2\gamma_{1}\gamma_{0}$ and
$\alpha_{0}\triangleq\gamma_{0}^{2}$. Since $l^{2}(i)=l(i)$, it
follows that 
\begin{equation}
\mathbb{E}\Big[\big(\prod_{i=j+1}^{t-1}(\gamma_{1}l(i)+\gamma_{0})\big)^{2}\Big]=\mathbb{E}\Big[\prod_{i=j+1}^{t-1}(\alpha_{1}l(i)+\alpha_{0})\Big].\label{eq:gamma-alpha-eq}
\end{equation}
By (\ref{eq:gamma_cond}), (\ref{eq:alpha_ineq}) holds. Therefore,
it follows from Lemma~\ref{LemmaDOne} that $\mathbb{E}\Big[\prod_{i=j+1}^{t-1}(\alpha_{1}l(i)+\alpha_{0})\Big]\leq\mu\theta^{(t-j-1)}$,
where $\mu\geq0$ and $\theta\in(0,1)$ are scalars that depend on
$\alpha_{1}$ and $\alpha_{0}$. Letting 
\begin{align}
f & \triangleq\mu^{\frac{1}{2}}/(1-\theta^{\frac{1}{2}}),\label{eq:FDef}
\end{align}
we obtain
\begin{align}
 & \sum_{j=0}^{t-2}\Big(\mathbb{E}\Big[\big(\prod_{i=j+1}^{t-1}(\gamma_{1}l(i)+\gamma_{0})\big)^{2}\Big]\Big)^{\frac{1}{2}}\leq\sum_{j=0}^{t-2}\mu^{\frac{1}{2}}\theta^{\frac{1}{2}(t-j-1)}\nonumber \\
 & \quad=\mu^{\frac{1}{2}}\sum_{i=1}^{t-1}\theta^{\frac{1}{2}i}\leq\mu^{\frac{1}{2}}\sum_{i=0}^{\infty}\theta^{\frac{1}{2}i}=f,\label{eq:expineq-1}
\end{align}
which completes the proof. \end{proof} 

Lemma~\ref{LemmaSumTwo} enables us to deal with quadratic terms
that involve the failure indicator $l(\cdot)$. We are now ready to
prove Theorem~\ref{TheoremDependentDisturbance}.

\emph{Proof of Theorem \ref{TheoremDependentDisturbance}}: By (\ref{eq:two-norm-ineq}),
\begin{align}
\mathbb{E}[\|x(t)\|_{2}] & \leq\frac{c_{2}}{c_{1}}\mathbb{E}\big[\prod_{i=0}^{t-1}(\zeta_{1}l(i)+\zeta_{0})\big]\|x_{0}\|_{2}\nonumber \\
 & \quad+\frac{c_{2}}{c_{1}}\sum_{j=0}^{t-2}\mathbb{E}\Big[\Big(\prod_{i=j+1}^{t-1}(\zeta_{1}l(i)+\alpha_{0})\Big)\|w(j)\|_{2}\Big]\nonumber \\
 & \quad+\frac{c_{2}}{c_{1}}\mathbb{E}[\|w(t-1)\|_{2}],\quad t\in\mathbb{N}.\label{eq:expectation_two-norm}
\end{align}
To show (\ref{eq:noise_to_state_three}), we obtain upper bounds for
the expectation terms on the right-hand side of (\ref{eq:expectation_two-norm})
by using Schwarz's and Jensen's inequalities. First, by Schwarz's
inequality (see Section~6.8 of \cite{williams2010probability}),
\begin{align}
 & \mathbb{E}\Big[\Big(\prod_{i=j+1}^{t-1}(\zeta_{1}l(i)+\zeta_{0})\Big)\|w(j)\|_{2}\Big]\nonumber \\
 & \quad\leq\Big(\mathbb{E}\big[\big(\prod_{i=j+1}^{t-1}(\zeta_{1}l(i)+\zeta_{0})\big)^{2}\big]\Big)^{\frac{1}{2}}\big(\mathbb{E}[\|w(j)\|_{2}^{2}]\big)^{\frac{1}{2}}.\label{eq:schwarz-result}
\end{align}
Furthermore, by Jensen's inequality (see Section~6.6 of \cite{williams2010probability}),
\begin{align}
\mathbb{E}[\|w(t-1)\|_{2}] & \leq\big(\mathbb{E}[\|w(t-1)\|_{2}^{2}]\big)^{\frac{1}{2}}.\label{eq:jensen-result}
\end{align}
As $\Big(\mathbb{E}\big[\big(\prod_{i=j+1}^{t-1}(\zeta_{1}l(i)+\zeta_{0})\big)^{2}\big]\Big)^{\frac{1}{2}}\geq0$,
we obtain from (\ref{eq:expectation_two-norm})--(\ref{eq:jensen-result}),
\begin{align}
 & \mathbb{E}[\|x(t)\|_{2}]\leq\frac{c_{2}}{c_{1}}\mathbb{E}\big[\prod_{i=0}^{t-1}(\zeta_{1}l(i)+\zeta_{0})\big]\|x_{0}\|_{2}\nonumber \\
 & \,\,\quad+\frac{c_{2}}{c_{1}}\Bigg(\sum_{j=0}^{t-2}\Big(\mathbb{E}\Big[\big(\prod_{i=j+1}^{t-1}(\zeta_{1}l(i)+\zeta_{0})\big)^{2}\Big]\Big)^{\frac{1}{2}}+1\Bigg)\nonumber \\
 & \,\,\qquad\cdot\max_{i\in\{0,\ldots,t-1\}}(\mathbb{E}[\|w(i)\|_{2}^{2}])^{\frac{1}{2}},\quad t\in\mathbb{N}.\label{eq:exp_ineq_in_noise_three}
\end{align}
First, we apply Lemma~\ref{LemmaDOne} to find an upper bound of
the term $\mathbb{E}\big[\prod_{i=0}^{t-1}(\zeta_{1}l(i)+\zeta_{0})\big]$
on the right-hand side of (\ref{eq:exp_ineq_in_noise_three}). To
this end let $\alpha_{1}\triangleq\zeta_{1}$ and $\alpha_{2}\triangleq\zeta_{2}$.
Since $\|A\|>1$, it follows from (\ref{eq:abkcond-difficult}) that
$\|A+BK\|<1$. As a result, $\alpha_{1}=\zeta_{1}>0$ and $\alpha_{0}=\zeta_{0}\in[0,1)$.
Furthermore, (\ref{eq:abkcond-difficult}) implies $(\zeta_{1}^{2}+2\zeta_{1}\zeta_{0})\hat{p}(\hat{\phi}_{\mathrm{J}})+\zeta_{0}^{2}<1$.
Using this inequality together with $\hat{p}(\hat{\phi}_{\mathrm{J}})\leq1$,
we obtain 
\begin{align*}
(\zeta_{1}\hat{p}(\hat{\phi}_{\mathrm{J}})+\zeta_{0})^{2} & =\zeta_{1}^{2}\hat{p}^{2}(\hat{\phi}_{\mathrm{J}})+2\zeta_{1}\zeta_{0}\hat{p}(\hat{\phi}_{\mathrm{J}})+\zeta_{0}^{2}\\
 & \leq(\zeta_{1}^{2}+2\zeta_{1}\zeta_{0})\hat{p}(\hat{\phi}_{\mathrm{J}})+\zeta_{0}^{2}<1,
\end{align*}
which implies (\ref{eq:alpha_ineq}). It then follows from Lemma~\ref{LemmaDOne}
that 
\begin{align}
\mathbb{E}\Big[\prod_{i=0}^{t-1}(\zeta_{1}l(i)+\zeta_{0})\Big] & \leq\mu\theta^{t},\label{eq:up1}
\end{align}
 where $\mu\geq0$ and $\theta\in(0,1)$ depend on $\zeta_{1}$ and
$\zeta_{0}$. Next, we apply Lemma~\ref{LemmaSumTwo} to find an
upper bound of the summation term $\sum_{j=0}^{t-2}\Big(\mathbb{E}\Big[\big(\prod_{i=j+1}^{t-1}(\zeta_{1}l(i)+\zeta_{0})\big)^{2}\Big]\Big)^{\frac{1}{2}}$.
Specifically, let $\gamma_{1}\triangleq\zeta_{1}$ and $\gamma_{0}\triangleq\zeta_{0}$.
By (\ref{eq:abkcond-difficult}), we have (\ref{eq:gamma_cond}).
Noting that $\gamma_{1}>0$ and $\gamma_{0}\in[0,1)$, we obtain by
Lemma~\ref{LemmaSumTwo} that 
\begin{align}
\sum_{j=0}^{t-2}\Big(\mathbb{E}\Big[\big(\prod_{i=j+1}^{t-1}(\zeta_{1}l(i)+\zeta_{0})\big)^{2}\Big]\Big)^{\frac{1}{2}} & \leq f,\label{eq:up2}
\end{align}
 where $f\geq0$ depends on $\zeta_{1}$ and $\zeta_{0}$. Now, by
letting 
\begin{align}
\hat{\theta}\triangleq\theta,\quad\hat{\mu} & \triangleq\mu c_{2}/c_{1},\quad\hat{f}\triangleq(f+1)c_{2}/c_{1},\label{eq:SecondThetaMuFHatDef}
\end{align}
we obtain (\ref{eq:noise_to_state_three}) from (\ref{eq:exp_ineq_in_noise_three})--(\ref{eq:up2}).
\hfill $\blacksquare$

Theorem~\ref{TheoremDependentDisturbance} shows
that if the jamming attacks satisfy Assumption~\ref{SecondAssumption}
with a sufficiently small $\hat{\phi}_{\mathrm{J}}$ such that (\ref{eq:abkcond-difficult})
holds, then the first moment of the state satisfies the bound in (\ref{eq:noise_to_state_three}). 

\begin{remark}The constants $\hat{\theta}$ and $\hat{\mu}$ of the
first-moment inequality (\ref{eq:noise_to_state_three}) are the same
as those provided in Remark~\ref{RemarkFirstHatValues} for the bounded-disturbance
case. Specifically, $\hat{\theta}$ and $\hat{\mu}$ are given respectively
by (\ref{eq:ThetaHatValue}) and (\ref{eq:MuHatValue}). Furthermore,
$\hat{f}\geq0$ in (\ref{eq:noise_to_state_three}) can be obtained
from (\ref{eq:FDef}) and (\ref{eq:SecondThetaMuFHatDef}) as 
\begin{align}
\hat{f}=\frac{c_{2}}{c_{1}}\Big(\frac{\|A\|{}^{(\tilde{T}^{*}-1)}\tilde{\theta}^{-(\tilde{T}^{*}-1)/2}}{1-\tilde{\theta}^{1/2}}+1\Big),\label{eq:FHatValue}
\end{align}
where $\tilde{\theta}\triangleq(1-\hat{p}(\frac{\hat{\kappa}}{\tilde{T}^{*}}+\hat{\phi}_{\mathrm{J}}))\|A+BK\|^{2}+\hat{p}(\frac{\hat{\kappa}}{\tilde{T}^{*}}+\hat{\phi}_{\mathrm{J}})\|A\|^{2}$,
and $\tilde{T}^{*}$ is a positive integer that satisfies $(1-\hat{p}(\frac{\hat{\kappa}}{\tilde{T}^{*}}+\hat{\phi}_{\mathrm{J}}))\|A+BK\|^{2}+\hat{p}(\frac{\hat{\kappa}}{\tilde{T}^{*}}+\hat{\phi}_{\mathrm{J}})\|A\|^{2}<1$.
\end{remark}

Theorem~\ref{TheoremDependentDisturbance} is applicable to scenarios
where the condition (\ref{eq:almost_sure_noise_bound}) of Theorem~\ref{TheoremBoundedDisturbance}
may fail to hold. In particular, if disturbance distributions have
infinite support, then (\ref{eq:almost_sure_noise_bound}) does not
hold (e.g., Gaussian distribution with $w(t)\sim\mathcal{N}(m,\Sigma)$
where $m\in\mathbb{R}^{n}$ and $\Sigma\in\mathbb{R}^{n\times n}$
is a positive-definite matrix). In such cases, Theorem~\ref{TheoremDependentDisturbance}
can be utilized. If $\mathbb{E}[\|w(t)\|_{2}^{2}]\leq\tilde{w}$ holds
for all $t\in\mathbb{N}_{0}$ with a scalar $\tilde{w}\geq0$, then
it follows from (\ref{eq:noise_to_state_three}) that $\limsup_{t\to\infty}\mathbb{E}[\|x(t)\|_{2}]\leq\hat{f}\tilde{w}^{\frac{1}{2}}$,
indicating the long-run boundedness of expected state norm. 

Although Theorem~\ref{TheoremDependentDisturbance} is applicable
to a wider range of scenarios in terms of the disturbance, the condition
(\ref{eq:abkcond-difficult}) is more restrictive than the condition
(\ref{eq:abkcond-vhat}) of Theorem~\ref{TheoremBoundedDisturbance}.
In particular, we have $\big((1-\hat{p}(\hat{\phi}_{\mathrm{J}}))\|A+BK\|+\hat{p}(\hat{\phi}_{\mathrm{J}})\|A\|\big)^{2}<(1-\hat{p}(\hat{\phi}_{\mathrm{J}}))\|A+BK\|^{2}+\hat{p}(\hat{\phi}_{\mathrm{J}})\|A\|^{2}$
for $\hat{p}(\hat{\phi}_{\mathrm{J}})\in(0,1)$ indicating that (\ref{eq:abkcond-difficult})
implies (\ref{eq:abkcond-vhat}), but not vice versa. We also note
that the finite second-moment condition in Theorem~\ref{TheoremDependentDisturbance}
holds in many control engineering scenarios. A particular example
is the Gaussian measurement noise setting. 

It is interesting that both Theorems~\ref{TheoremBoundedDisturbance}
and \ref{TheoremDependentDisturbance} can be used for assessing stability
in the scenarios where the transmission failure indicator process
$\{l(t)\}_{t\in\mathbb{N}_{0}}$ and the disturbance process $\{w(t)\}_{t\in\mathbb{N}_{0}}$
are not independent of each other. This is the case, e.g., when the
state measurements received by the controller are subject to noise.
Note also that the conditions in both theorems can be checked using
different induced matrix norms $\|\cdot\|$. Certain norms can provide
less conservative results, as illustrated in Section~\ref{sec:Illustrative-Numerical-Examples}. 

\section{Numerical Example}

\label{sec:Illustrative-Numerical-Examples}

Consider the networked control system (\ref{eq:closed-loop-system})
with 
\begin{align*}
A=\left[\begin{array}{cc}
0.1 & -1\\
1.1 & 1.8
\end{array}\right],\, & B=\left[\begin{array}{c}
0\\
1
\end{array}\right],\,K=[-0.9277\,-1.2615],
\end{align*}
 and  the channel model from Example~\ref{ExampleOutage} with outage
probability $p$ given by (\ref{eq:p-outage}), where $\xi=2$, $\sigma=0.05$,
$\underline{\gamma}=2$, and $b_{1}=b_{2}=1$. 

We first investigate the disturbance-free case ($w(t)\equiv0$). Noting
that $p$ is a concave, continuous, and nondecreasing function, we
set $\hat{p}(\phi)\triangleq p(\phi)$, which satisfies (\ref{eq:p_phat_ineq}).
By Theorem~3.5 of \cite{ahmetifacwc2017}, the system is almost surely
asymptotically stable under any attacks that satisfy Assumption~\ref{FirstAssumption}
with $\overline{\phi}_{\mathrm{J}}\leq0.62$. The analysis in \cite{ahmetifacwc2017}
is Lyapunov-based, and for the case with $\overline{\phi}_{\mathrm{J}}=0.62$,
it uses the Lyapunov-like function $V(x)\triangleq x^{\mathrm{T}}Px$
with the positive-definite matrix 
\begin{align}
P & =\left[\begin{array}{cc}
0.7728 & 0.8554\\
0.8554 & 3.2649
\end{array}\right].\label{eq:example-p-matrix}
\end{align}
 The matrix $P$ is also useful for the first-moment stability analysis.
In particular, we can use the matrix norm $\|\cdot\|$ induced by
the vector norm $\|x\|_{P}\triangleq\sqrt{x^{\mathrm{T}}Px}$. By
using this matrix norm, the stability condition in Proposition~A.1
 is satisfied for $\overline{\phi}_{\mathrm{J}}\leq0.27$. This indicates
that the networked control system (\ref{eq:closed-loop-system}) without
disturbance is first-moment geometrically stable under jamming attacks
that satisfy Assumption~\ref{FirstAssumption} with $\overline{\phi}_{\mathrm{J}}\leq0.27$.
Hence, in the disturbance-free case $\mathbb{E}[\|x(t)\|_{2}]$ converges
to zero with a geometric rate. The choice of the matrix norm is important
for stability analysis. For instance, in this example, the stability
condition in Proposition~A.1 does not hold with matrix norms induced
by $1$-norm, Euclidean-norm, or infinity norm, because for those
norms, $\|A+BK\|>1$. 

\begin{figure}[t]
\centering  \includegraphics[width=0.9\columnwidth]{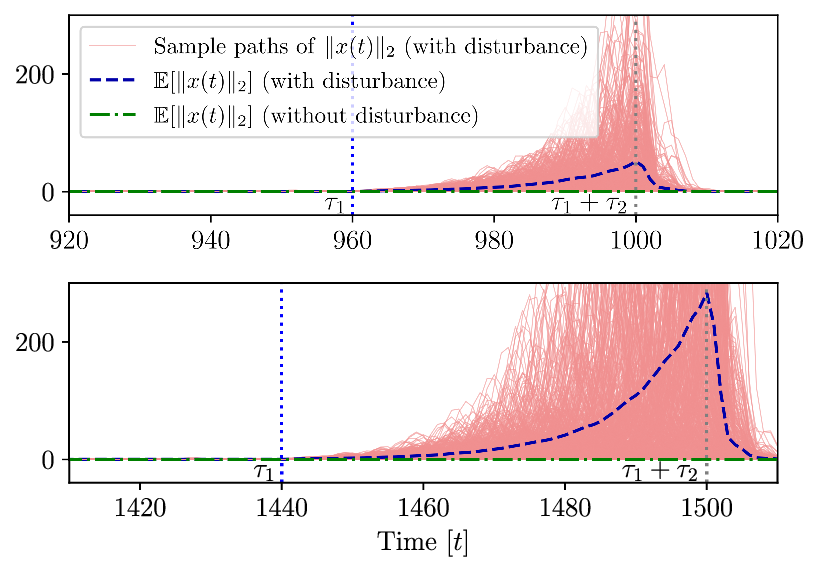} 

\caption{Effect of jamming attacks following (\ref{eq:ExampleStrategyV}) with
and without disturbance (Top: $\tau_{1}=960$, $\tau_{2}=40$, $\phi_{\mathrm{J}}^{*}=6.75$;
Bottom: $\tau_{1}=1440$, $\tau_{2}=60$, $\phi_{\mathrm{J}}^{*}=6.75$).
Approximate first moments $\mathbb{E}[\|x(t)\|_{2}]$ are obtained
through $500$ simulations with $x_{0}=[1,\,1]^{\mathrm{T}}$. }
 \label{Flo:Example1} 
\end{figure}

\subsubsection{Disturbance-free scenario}

As discussed in Section~\ref{subsec:Joint-Effect-of}, Assumption~\ref{FirstAssumption}
allows the attacker to jam the channel with very large interference
powers after waiting without attacking for sufficiently long durations.
For instance, for the attack strategy considered in (\ref{eq:ExampleStrategyV})
with $\tau_{1}=960$, $\tau_{2}=40$, and $\phi_{\mathrm{J}}^{*}=6.75$,
Assumption~\ref{FirstAssumption} is satisfied with $\overline{\kappa}=0$
and $\overline{\phi}_{\mathrm{J}}=0.27$. In the disturbance-free
case, this attack strategy does not create a problem for stability
since $\overline{\phi}_{\mathrm{J}}$ is sufficiently small. In particular,
after the long duration $\tau_{1}$ without attacks, the state norm
gets very close to zero, and as a result, the state norm after the
attack period of $\tau_{2}$ time steps is also small. 

\subsubsection{Scenarios with disturbance}

By contrast, in the case with disturbance, the attack strategy (\ref{eq:ExampleStrategyV})
makes the state norm grow at time $\tau_{1}+\tau_{2}$. This is because,
even after the long attack-free duration, the state norm cannot get
close to zero due to the disturbance. This is shown in the top part
of Fig.~\ref{Flo:Example1} with the disturbance given by 
\begin{align}
w(t) & =[\begin{array}{cc}
\cos(\vartheta)\widetilde{w}(t) & \sin(\vartheta)\widetilde{w}(t)\end{array}]^{\mathrm{T}},\quad t\in\mathbb{N}_{0},\label{eq:disturbancevartheta}
\end{align}
 where $\vartheta=\pi/2$ and $\widetilde{w}(t)\in\mathbb{R}$ at
each time $t$ is uniformly distributed in $[-0.5,0.5]$. Under disturbance,
the length $\tau_{2}$ of the attack period directly affects the growth
of the state norm. The attacker can increase the waiting time $\tau_{1}$
to attack with a longer duration $\tau_{2}$ with the same high interference
power $\phi_{\mathrm{J}}^{*}$ to make the state norm grow, while
still satisfying Assumption~\ref{FirstAssumption}. 

In the bottom part of Fig.~\ref{Flo:Example1}, we see that for the
same disturbance but with $\tau_{2}=60$, the state is driven to larger
values. Notice that with $\tau_{1}=1440$, $\tau_{2}=60$, and $\phi_{\mathrm{J}}^{*}=6.75$,
Assumption~\ref{FirstAssumption} is also satisfied with $\overline{\kappa}=0$
and $\overline{\phi}_{\mathrm{J}}=0.27$. Although after the time
$\tau_{1}+\tau_{2}$, the effect of the attack diminishes, the attacker
can repeat cycles of sleeping and jamming, and the state norm may
grow if the attacker uses higher interference powers for longer durations.
To guarantee a predetermined bound on the expected state norm, interference
power levels need to be restricted. This is achieved by Assumption~\ref{SecondAssumption}.
Under Assumption~\ref{SecondAssumption}, the attacker can attack
with a jamming interference power $\phi_{\mathrm{J}}^{*}>\hat{\phi}_{\mathrm{J}}$
consecutively for at most $\lfloor\hat{\kappa}/(\phi_{\mathrm{J}}^{*}-\hat{\phi}_{\mathrm{J}})\rfloor$
time steps. For instance, with $\hat{\kappa}=259.2$ and $\hat{\phi}_{\mathrm{J}}=0.27$,
the jamming attacks in the top part of Fig.~\ref{Flo:Example1} satisfy
Assumption~\ref{SecondAssumption}. However, the jamming attacks
in the bottom part do not satisfy Assumption~\ref{SecondAssumption}
with the same $\hat{\kappa}$ and $\hat{\phi}_{\mathrm{J}}$ due to
the longer attack duration. For a duration of $60$ time steps, the
maximum allowed interference power is $\phi_{\mathrm{J}}^{*}=4.59$.
We remark that the parameters $\hat{\kappa}$ and $\hat{\phi}_{\mathrm{J}}$
can be selected to reflect the capabilities of the attacker. 

If the jamming strategy satisfies Assumption~\ref{SecondAssumption}
with $\hat{\phi}_{\mathrm{J}}\leq0.27$, then by Theorem~\ref{TheoremBoundedDisturbance},
the first moment of the state satisfies the bound in (\ref{eq:noise_to_state_one})
for any bounded disturbance. If the disturbance is not bounded, then
Theorem~\ref{TheoremDependentDisturbance} can be applied; by Theorem~\ref{TheoremDependentDisturbance},
the bound in (\ref{eq:noise_to_state_three}) holds if the attacker
is less powerful with $\hat{\phi}_{\mathrm{J}}\leq0.1$.

The first moment bound provided in (\ref{eq:noise_to_state_one})
can be evaluated using the values of $\hat{\theta}$, $\hat{\mu}$,
and $\hat{d}$ given in Remark~\ref{RemarkFirstHatValues}. We note
that the bound is not tight for this example. This is partly because
the inequalities (\ref{eq:two-norm-ineq}) and (\ref{eq:expectation_bounded_noise_case}),
which relate packet transmission failure indicators to the norm of
the state and its first-moment, are not tight for multi-dimensional
systems. Another factor is that the disturbance in this example is
a stochastic process and does not necessarily increase the state norm
at each time. 

In Figs.~\ref{Flo:BoundFig1} and \ref{Flo:BoundFig2}, we show plots
for the first moment bound in (\ref{eq:noise_to_state_one}) with
the values of $\hat{\theta}$, $\hat{\mu}$, and $\hat{d}$ provided
in Remark~\ref{RemarkFirstHatValues}. These plots show that the
bound becomes larger for more powerful attacks. Notice that in Fig.~\ref{Flo:BoundFig1}
the bounds are not visibly decreasing. This is because $\hat{\theta}\in(0,1)$
is close to $1$ and the term $\hat{\mu}\hat{\theta}^{t}\|x_{0}\|_{2}$
is smaller than $\hat{d}\overline{w}$. 
\begin{figure}[t]
\centering  \includegraphics[width=0.9\columnwidth]{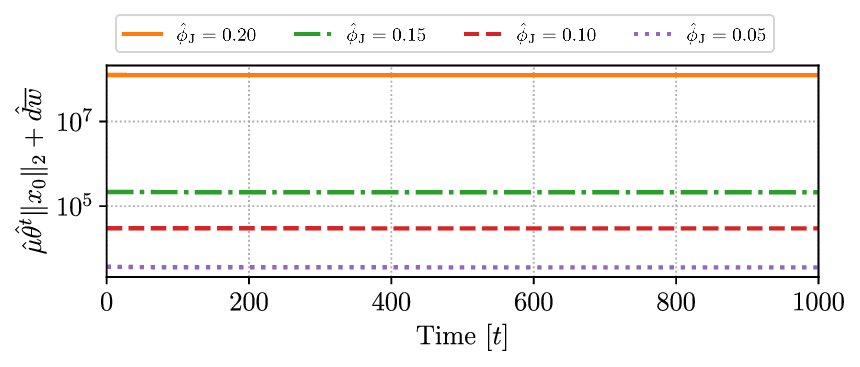} 

\caption{First moment bound in (\ref{eq:noise_to_state_one}) with respect
time for different values of the attacker's average interference power
$\hat{\phi}_{\mathrm{J}}$ and $\hat{\kappa}=1$.}
 \label{Flo:BoundFig1} 
\end{figure}

\begin{figure}[t]
\centering  \includegraphics[width=0.9\columnwidth]{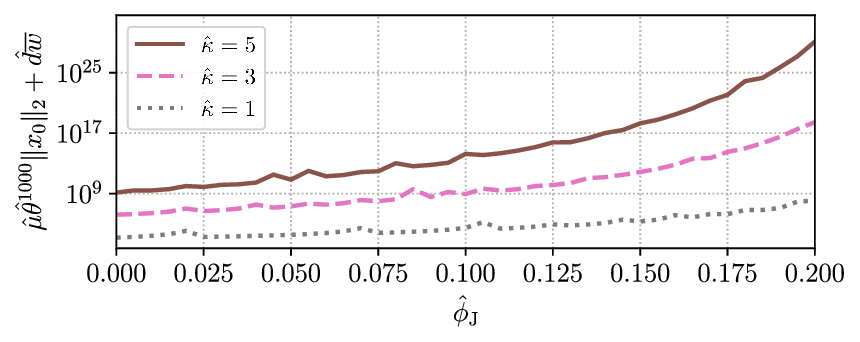} 

\caption{First moment bound in (\ref{eq:noise_to_state_one}) evaluated at
time $t=1000$ with respect to the attacker's average interference
power $\hat{\phi}_{\mathrm{J}}$ for different values of $\hat{\kappa}$.}
 \label{Flo:BoundFig2} 
\end{figure}

\begin{figure}[t]
\centering  \includegraphics[width=0.9\columnwidth]{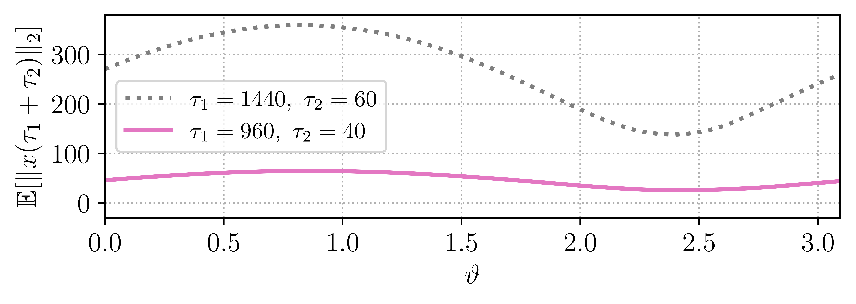} 

\caption{Performance index $\mathbb{E}[\|x(\tau_{1}+\tau_{2})\|_{2}]$ for
different values of $\vartheta$ in (\ref{eq:disturbancevartheta})
and different settings of $\tau_{1},\tau_{2}$ in (\ref{eq:ExampleStrategyV}).}
 \label{Flo:wrtTheta} 
\end{figure}

Next, we explore the effects of different disturbance realizations.
In particular, we run simulations for different values of $\vartheta\in[0,\pi]$
in (\ref{eq:disturbancevartheta}) and evaluate $\mathbb{E}[\|x(\tau_{1}+\tau_{2})\|_{2}]$
as a performance index. The evaluations are done under the attack
strategy (\ref{eq:ExampleStrategyV}). This strategy was also used
for obtaining Fig.~\ref{Flo:Example1} for the particular value $\vartheta=\pi/2$
corresponding to the situation where the disturbance only affects
the second state. Different values of $\vartheta$ result in different
levels of disturbance on the first and the second states. Fig.~\ref{Flo:wrtTheta}
shows that $\mathbb{E}[\|x(\tau_{1}+\tau_{2})\|_{2}]$ can vary largely
depending on $\vartheta$ even though the disturbance magnitude $\|w(t)\|=|\widetilde{w}(t)|$
does not depend on how $\vartheta$ is chosen. Here the value of $\mathbb{E}[\|x(\tau_{1}+\tau_{2})\|_{2}]$
is $\pi$-periodic, because the distribution of disturbance as a function
of $\vartheta$ is $\pi$-periodic. Variations in $\mathbb{E}[\|x(\tau_{1}+\tau_{2})\|_{2}]$
indicate that jamming attacks can be more/less effective depending
on how the disturbance enters in the dynamics.

\begin{figure}[t]
\centering  \includegraphics[width=0.9\columnwidth]{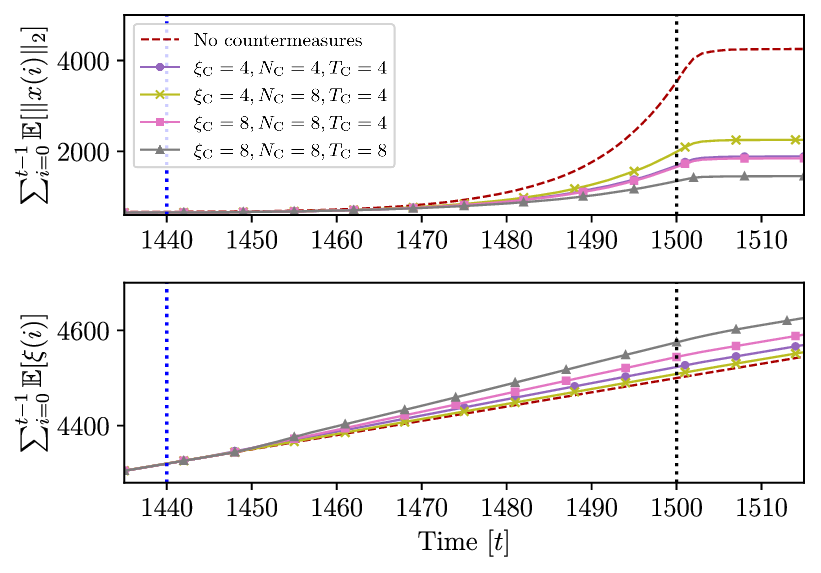} 

\caption{Approximate expected values of total state norm (top) and total transmission
power (bottom) for different countermeasure parameters compared to
the case with no countermeasures from Fig.~\ref{Flo:Example1}.}
 \label{Flo:Countermeasures} 
\end{figure}

\subsubsection{Countermeasures against jamming}

The damaging effects of the jamming attacks can be reduced by adjusting
the transmission power ($\xi$ in (\ref{eq:p-outage})). Consider
the setup where successful communications are replied with acknowledgement
messages. In the case where communication fails, the controller would
receive no acknowledgement, which indicates the failure. In this setup,
the controller can improve the overall performance by increasing the
transmission power when there are many consecutive failures. This
countermeasure against jamming can be described as follows. If $l(t-i)=1$
for each $i\in\{1,\ldots,N_{\mathrm{C}}\}$ (representing $N_{\mathrm{C}}$
total consecutive failures), then at time $t$, the transmission power
$\xi(t)$ is set to a value $\xi_{\mathrm{C}}$ (larger than the nominal
value $3$ used above) for a duration of $T_{\mathrm{C}}$ time steps.
Thus, at those time steps, failures become less likely. After $T_{\mathrm{C}}$
time steps, the transmission power is set back to its nominal (lower)
value and the countermeasure system restarts counting consecutive
failures. We explore the effectiveness of this countermeasure against
the attacks that are illustrated in the bottom part of Fig.~\ref{Flo:Example1}.
Specifically, Fig.~\ref{Flo:Countermeasures} shows the expected
total state norm ($\sum_{i=0}^{t-1}\mathbb{E}[\|x(i)\|_{2}]$) and
expected total transmission power ($\sum_{i=0}^{t-1}\mathbb{E}[\xi(i)]$)
approximated through 500 simulations for different parameter values
$\xi_{\mathrm{C}}\in\{4,8\},N_{\mathrm{C}}\in\{4,8\},T_{\mathrm{C}}\in\{4,8\}$.
The results indicate that the effects of jamming can be mitigated
by temporarily increasing transmission powers, and the performance
gets better with larger total transmission power use.

\section{Conclusion}

\label{sec:Conclusion} 

We explored the networked control problem under jamming attacks with
time-varying interference power. Specifically, we investigated the
effects of jamming attacks on systems that are subject to disturbance,
and obtained conditions under which the first moment of the state
stays bounded. Our results indicate that if the disturbance is known
to be bounded, stability of a system can be guaranteed under larger
average jamming interference powers. 

Our results can be extended for the case where multiple wireless channels
are used for the transmission of state and control data. In such cases,
increasing the number of channels through which the plant and the
controller communicate can increase the level of tolerance against
certain jamming attack scenarios.

One of our future research directions is to increase robustness properties
of the overall system by utilizing predictive control approaches proposed
previously in \cite{feng2017resilient}. Another future work is to
provide an analysis of the networked control system under time-varying
transmission powers. In this line of research, for a wireless networked
control problem without attacks, \cite{varma2017stochastic} recently
explored stability and energy-efficiency under time-varying transmission
powers. 

\bibliographystyle{ieeetr}
\bibliography{references}

\appendix 

Here we provide an analysis of moment stability of the networked control
system under Assumption~\ref{FirstAssumption}. In particular, the
following result provides a condition under which the first-moment
of the state ($\mathbb{E}[\|x(t)\|_{2}]$) of system (\ref{eq:closed-loop-system})
converges to zero at a geometric rate. 

\begin{approp} \label{PropositionMomentConvergence}Consider the
closed-loop networked control system (\ref{eq:closed-loop-system})
for the case where $w(t)=0$, $t\in\mathbb{N}_{0}$. Suppose that
the attacker's interference power process $\{\phi_{\mathrm{J}}(t)\in[0,\infty)\}_{t\in\mathbb{N}_{0}}$
satisfies Assumption~\ref{FirstAssumption}. Moreover, assume
\begin{align}
 & (1-\hat{p}(\overline{\phi}_{\mathrm{J}}))\|A+BK\|+\hat{p}(\overline{\phi}_{\mathrm{J}})\|A\|<1.\label{eq:abkcond}
\end{align}
 Then the closed-loop system (\ref{eq:closed-loop-system}) is first-moment
geometrically stable, that is, there exist $\overline{\mu}\geq0$
and $\overline{\theta}\in(0,1)$ such that 
\begin{align}
\mathbb{E}[\|x(t)\|_{2}] & \leq\overline{\mu}\overline{\theta}^{t}\|x_{0}\|_{2},\quad t\in\mathbb{N}.\label{eq:moment-convergence}
\end{align}
 \end{approp} \medskip

The proof of Proposition~\ref{PropositionMomentConvergence} is based
on the following result. 

\begin{aplemma} \label{LemmaDTwo} Suppose that the attacker's interference
power process  $\{\phi_{\mathrm{J}}(t)\in[0,\infty)\}_{t\in\mathbb{N}_{0}}$
satisfies Assumption~\ref{FirstAssumption}. Then for every $\alpha_{1}\geq0$,
$\alpha_{0}\geq0$ that satisfy
\begin{align}
\alpha_{1}\hat{p}(\overline{\phi}_{\mathrm{J}})+\alpha_{0} & <1,\label{eq:alpha_ineq-1}
\end{align}
there exist scalars $\mu\geq0$ and $\theta\in(0,1)$ such that 
\begin{align}
\mathbb{E}[\prod_{i=0}^{t-1}(\alpha_{1}l(i)+\alpha_{0})] & \leq\mu\theta^{t},\quad t\in\mathbb{N}.\label{eq:geometric-lemma-result-1}
\end{align}
 \end{aplemma} 

\begin{proof} The proof is similar to that of Lemma~\ref{LemmaDOne}.
In particular, we have (\ref{eq:expxhineq-1-1}) and (\ref{eq:powerbound})
with $t_{1}=0$, $t_{2}=t$, $h(v)\triangleq\alpha_{1}\hat{p}(v)+\alpha_{0}$,
and hence, $\mathbb{E}[\prod_{i=0}^{t-1}(\alpha_{1}l(i)+\alpha_{0})]\leq\mathbb{E}[h^{t}(\frac{1}{t}\sum_{i=0}^{t-1}\phi_{\mathrm{J}}(i))]$.
By Assumption~\ref{FirstAssumption}, we then obtain $\mathbb{E}[h^{t}(\frac{1}{t}\sum_{i=0}^{t-1}\phi_{\mathrm{J}}(i))]\leq h^{t}(\frac{\overline{\kappa}}{t}+\overline{\phi}_{\mathrm{J}})$.
Therefore, by (\ref{eq:alpha_ineq-1}), after letting $T^{*}$ be
a positive integer such that $h(\frac{\overline{\kappa}}{T^{*}}+\overline{\phi}_{\mathrm{J}})<1$
and defining
\begin{align}
\theta\triangleq h(\frac{\overline{\kappa}}{T^{*}}+\overline{\phi}_{\mathrm{J}}),\quad & \mu\triangleq(\alpha_{1}+\alpha_{0})^{T^{*}-1}\theta^{-(T^{*}-1)},\label{eq:thetamuappendix}
\end{align}
we obtain (\ref{eq:geometric-lemma-result-1}). \end{proof}

\emph{Proof of Proposition \ref{PropositionMomentConvergence}}: By
(\ref{eq:two-norm-ineq}) with $w(t)=0$, $t\in\mathbb{N}_{0}$, we
have
\begin{align}
\mathbb{E}[\|x(t)\|_{2}] & \leq\frac{c_{2}}{c_{1}}\mathbb{E}\Big[\prod_{i=0}^{t-1}(\zeta_{1}l(i)+\zeta_{0})\Big]\|x_{0}\|_{2}.\label{eq:two-norm-ineq-1-1}
\end{align}
 Next, we apply Lemma~\ref{LemmaDTwo}. First, since $\|A\|>1$,
(\ref{eq:abkcond}) implies $\|A+BK\|\in[0,1)$, and thus, $\zeta_{1}>0$,
$\zeta_{0}\in[0,1)$. With $\alpha_{1}=\zeta_{1}$ and $\alpha_{0}=\zeta_{0}$,
(\ref{eq:abkcond}) implies (\ref{eq:alpha_ineq-1}). Therefore, by
Lemma~\ref{LemmaDTwo}, we have $\mathbb{E}\big[\prod_{i=0}^{t-1}(\zeta_{1}l(i)+\zeta_{0})\big]\leq\mu\theta^{t}$,
where $\mu\geq0$ and $\theta\in(0,1)$. Hence, by (\ref{eq:two-norm-ineq-1-1}),
the inequality (\ref{eq:moment-convergence}) holds with $\overline{\mu}\triangleq\frac{c_{2}}{c_{1}}\mu$
and $\overline{\theta}=\theta$. \hfill $\blacksquare$

Proposition~\ref{PropositionMomentConvergence} provides a method
to check the first-moment geometric stability of the system (\ref{eq:system}),
(\ref{eq:control-input}) under jamming attacks that satisfy Assumption~\ref{FirstAssumption}.
The scalars $\|A+BK\|$ and $\|A\|$ in condition (\ref{eq:abkcond})
respectively represent the behavior of the closed-loop dynamics under
successful transmissions and the open-loop dynamics under failed transmissions.
Here, we select the matrix norm $\|\cdot\|$ to ensure $\|A+BK\|<1$.
This is possible since the feedback gain $K$ is designed to make
$A+BK$ a Schur matrix, for which such a matrix norm can be constructed
(see Corollary 9.3.4 of \cite{hornmatrixanalysis}). On the other
hand, for unstable open-loop dynamics, we have $\|A\|>1$. Notice
that the inequality in (\ref{eq:abkcond}) holds if the upper bound
$\overline{\phi}_{\mathrm{J}}$ of the average jamming interference
power is sufficiently small so that $\hat{p}(\overline{\phi}_{\mathrm{J}})$
is sufficiently close to zero. In such cases, transmission failures
happen sufficiently rarely in average, and thus the overall networked
control system frequently follows the stable behavior of the closed-loop
dynamics and the geometric convergence of the first-moment of the
state as in (\ref{eq:moment-convergence}) can be guaranteed. As we
establish in the proof, the scalar $\overline{\theta}$ in (\ref{eq:moment-convergence})
represents the rate of convergence, and it depends on $\hat{p}(\overline{\phi}_{\mathrm{J}})$
as well as the scalars $\|A+BK\|$ and $\|A\|$. In particular, if
the bound $\overline{\phi}_{\mathrm{J}}$ on the long run average
jamming interference power is small, then $\overline{\theta}$ is
also small, indicating faster convergence of the first-moment.

\emph{First-moment geometric stability} discussed in Proposition~\ref{PropositionMomentConvergence}
is a stronger notion of stochastic stability in comparison to \emph{almost-sure
asymptotic stability} explored in \cite{ahmetifacwc2017}. As expected,
first-moment geometric stability condition (\ref{eq:moment-convergence})
is more restrictive with respect to the attack parameter $\overline{\phi}_{\mathrm{J}}$.
Specifically, the almost-sure asymptotic stability condition presented
in Theorem~3.5 of \cite{ahmetifacwc2017} reduces to 
\begin{align}
(1-\hat{p}(\overline{\phi}_{\mathrm{J}}))\ln\|A+BK\|+\hat{p}(\overline{\phi}_{\mathrm{J}})\ln\|A\| & <0,\label{eq:almost-sure}
\end{align}
with $\left\Vert \cdot\right\Vert $ denoting the matrix norm induced
by the vector norm $\|x\|_{P}\triangleq\sqrt{x^{\mathrm{T}}Px}$ where
$P\in\mathbb{R}^{n\times n}$ is a positive-definite matrix. For this
matrix norm, (\ref{eq:abkcond}) implies (\ref{eq:almost-sure}).

\end{document}